\documentclass[12pt,reqno]{amsart}
\usepackage{amsmath}%
\usepackage{amsfonts}%
\usepackage{amssymb}%
\usepackage{graphicx}
\usepackage{subcaption}
\usepackage{natbib}
\usepackage{bbm}
\usepackage{enumerate}
\usepackage{xcolor}

\usepackage{multirow}
\usepackage[margin=2.4cm]{geometry} 
\usepackage[misc]{ifsym}
\newtheorem{theorem}{Theorem}
\theoremstyle{plain}

\newtheorem{corollary}[theorem]{Corollary}

\newtheorem{definition}[theorem]{Definition}
\newtheorem{example}[theorem]{Example}

\newtheorem{lemma}[theorem]{Lemma}

\newtheorem{proposition}[theorem]{Proposition}
\newtheorem{remark}[theorem]{Remark}

\newtheorem{assumption}[theorem]{Assumption}
\numberwithin{equation}{section}
\numberwithin{theorem}{section}

\newcommand{\m}{\text{-}}
\newcommand{\Q}{\mathbb{Q}}
\newcommand{\p}{\mathbb{P}}

\newcommand{\cQ}{\mathcal{Q}}
\newcommand{\cl}{\text{cl}}
\newcommand{\conv}{\text{conv}}
\DeclareMathOperator{\esssup}{ess\ sup}
\DeclareMathOperator{\essinf}{ess\ inf}
\newcommand{\I}{\mathbb{R} \backslash \{0\}}

\begin{document}
\title{Robustness of delta hedging in a jump-diffusion model}
\author[F. Bosserhoff and M. Stadje]{Frank Bosserhoff*, Mitja Stadje*}
\address[]
{*Institute of Insurance Science and Institute of Financial Mathematics, Ulm University \newline%
\indent Helmholtzstrasse 20, 89081 Ulm, Germany}%
\email[]{\Letter\ frank.bosserhoff@uni-ulm.de, mitja.stadje@uni-ulm.de}%

\date{\today}

\begin{abstract} 
Suppose an investor aims at Delta hedging a European contingent claim $h(S(T))$ in a jump-diffusion model, but incorrectly specifies the stock price's volatility and jump sensitivity, so that any hedging strategy is calculated under a misspecified model. When does the erroneously computed strategy approximate the true claim in an appropriate sense? If the misspecified volatility and jump sensitivity dominate the true ones, we show that following the misspecified Delta strategy does super-replicate $h(S(T))$ in expectation among a wide collection of models. We also show that if a robust pricing operator with a whole class of models is used, the corresponding hedge is dominating the contingent claim under each model in expectation. Such a hedging error is also called a \emph{good-deal} or a $\rho$\emph{-arbitrage}. In the pure Poisson case, we establish that an overestimation of the jump sensitivity results in an almost sure one-sided hedge. Moreover, in general the misspecified price of the option dominates the true one if the volatility and the jump sensitivity are overestimated. Our results rely on proving stochastic flow properties of the jump-diffusion and the convexity of the contingent claim's value function.
\end{abstract}
\maketitle

\section{Introduction}

In the scientific literature on pricing and hedging of contingent claims as well as in practice, one oftentimes presumes against better knowledge that some stochastic model mimics the developments at the stock market appropriately.  However, even if an investor is aware of the model type (for example Markovian diffusion, jump-diffusion, infinite activity pure jump process, etc.), functions such as the drift, the volatility as well as the jump sensitivity are to be specified. Determining an allocation rule based on misspecified model parameters and trading in the \emph{true} stock may result in a severe violation of the hedging objective. As Delta strategies yield a perfect hedge in complete markets, applying them in incomplete markets is tempting as well and often done in practice. This gives rise to studying their general robustness properties. In this paper we closely look at the performance of Delta strategies in jump-diffusion financial markets and provide sufficient conditions under which a Delta hedge yields a hedging error of some European contingent claim $h(S(T))$ which is either zero or a \emph{good-deal} (also called a $\rho$-arbitrage). Denote the set of financial positions which are acceptable by $\mathcal{A}$ and the set of hedging oppertunies (i.e., of admissible portolios for a self-financing strategy with zero initial capital) by $\mathcal{V}$. In the theory of good-deal pricing, the set $\mathcal{A}$ is interpreted to besides zero contain all financial positions which are ``too good to be true for a price of zero''. Then we can define a risk of a financial position $X$, say $\rho(X)$, as the minimal  capital amount needed to finance a hedging strategy such that the position becomes acceptable, i.e.,
\begin{equation}
\label{gooddeal}
\rho(X):=\inf\{m\in\mathbb{R}|\exists P_T \in \mathcal{V}:\,\,\, m+P_T+X\in \mathcal{A}\}.
\end{equation}
A good-deal arises if the price of $X$ is not in the interval $[-\rho(X),\rho(-X)]$ as then one may buy or sell $X$ and create a hedging portfolio with an overall initial gain (i.e., with a negative price) while obtaining a terminal value in $\mathcal{A}$.

 Good-deals or good-deal bounds have been introduced in \cite{bernardo2000gain}, \cite{gdb}, \cite{jaschke2001coherent}, and \cite{vcerny2002theory}. They have then been extended in various works, see for instance \cite{staum2004fundamental}, \cite{slinko}, \cite{cherny2008pricing}, \cite{arai2011good}, \cite{arai2014convex}, \cite{bielecki2015dynamic}, \cite{kratschmer2018optimal}, \cite{bion2020fully}, and are also called ``stricly acceptable opportunities'' in \cite{carr2001pricing}. They are called $\rho$-arbitrage in \cite{armstrong2019ineffectiveness}. Good-deal bounds have been introduced to narrow the no-arbitrage bounds (which rely on a.s. riskless profits) arguing that although not being arbitrage opportunities in the classical sense, good-deals are ``too good to be true'' and therefore should disappear in a competititve market. 
 \par
In addition to assuming a deterministic interest rate and that the misspecified stock price process is Markov, the fundamental assumption enabling the variety of results obtained in this paper is the convexity of the payoff function $h,$ which allows to prove that the option price function is convex in the current stock level. For a diffusion setting, \cite{kjs} show this property employing the theory of stochastic flows; \cite{hobson1998volatility} provides a simplification using coupling techniques and \cite{ekstrom2005superreplication} investigate convexity properties when the claim depends on several underlying assets.\par
 For example, the convexity of the option value function enables us to deduce an ordering result. To be more precise, we show that if the model volatility and jump sensitivity systematically overestimate the true ones, the model option price dominates the corresponding market option price. A general comparison result for the solution of one-dimensional stochastic differential equations (SDEs) can be found in \cite{peng2006necessary}. Ordering results by deriving sufficient conditions for the convexity of Euler schemes, generalizations to multi-dimensional special semimartingales and path-dependent options are to be found in \cite{bergenthum2006comparison,bergenthum2007comparison,bergenthum2008comparison}. Extensions to more general discretization schemes and applications to Bermudan option prices are discussed in \cite{pages2016convex}. Predictable representation results are used in \cite{arnaudon2008convex}. \cite{hobson2010comparison} employs coupling arguments to draw comparisons among option prices in various stochastic volatility models, see also \cite{criens2019couplings}. For convex ordering results with pathwise \emph{It\^{o} calculus} see \cite{kopfer2019comparison}.  A general overview of the impact of model uncertainty on pricing of contingent claims is provided in \cite{cont2006model}. \par
However, contrary to most of the above literature, we are mainly interested in the robustness of Delta hedging strategies. We show that the hedging error of the self-financing Delta hedging portfolio in an incomplete market is a good-deal or a $\rho$-arbitrage with respect to robustified expectations containing the benchmark models as soon as the misspecified volatility and jump sensitivity dominate the true ones. 
In local volatility models a $\Delta$-strategy yields an a.s. perfect hedge; consequently, \cite{kjs} call a Delta strategy robust if the physical Delta strategy is an a.s. superhedge for the claim as soon as the model volatility systematically overestimates the market volatility so that the hedging error not only yiels a good-deal but actually an arbitrage opportunity. \cite{schiedchef} establish the robustness of the Delta hedging strategy for general path-dependent options in local volatility models. They show that a sufficient condition for the robustness of the Delta strategy in every local volatility model is the directional convexity of the payoff function. \cite{gapeev2011robust} investigate  the robust hedging problem  when log-returns of the stock price are Gaussian and self-similar and the investor is not sure whether the market is efficient. In incomplete markets perfect hedging is not possible and consequently besides Delta hedging other methods like superhedging, mean-variance hedging, utility indifference pricing, or quantile hedging are considered. Except for superhedging, which is rather expensive, all these methods allow for a hedging error.

Throughout this paper, we consider an investor modeling the stock price based on L\'evy jump processes, but also allowing for a Brownian component. Such models are known to be more appropriate for mimicking the true stock price development as observed at the stock exchange than pure Brownian models, see for instance \cite{cont} as well as \cite{ct} and references therein. A good overview of the theory of L\'evy processes can be found in \cite{applebaum}, \cite{bertoin} and \cite{sato}, while \cite{barndorff} and \cite{kyprianou} focus on its applications in finance.\par
This paper is structured as follows. In Section 2 the basic model is provided. In Section 3 the main result is Theorem \ref{thm_convex} stating the convexity of the contingent claim's price. The major part of Section 3 is devoted to its proof. Section 4 shows that the Delta hedge leads to a good deal for a risk measure based on a robustification of the expectation, if the true volatility and jump sensitivity are dominated. It is also shown that the misspecified option's price is larger than the true one. Section 5 starts by introducing general robust pricing operators and discusses several examples. After proving a non-smooth version of \emph{It\^{o}'s lemma} for the case of finite jump activity, it is established that in each model the hedging error induced by the Delta strategy is a submartingale, which allows deducing robustness and general good-deal properties. In Section 6, we consider a.s. superhedges and argue that under Markovian assumptions the payoff is not necessarily linear only if the jump sensitivity is independent of the jump size and the stock price process is only driven by either the jump process \emph{or} the diffusion component. For the latter case the robustness of the replicating strategy is exemplified in the pure Poisson case. We remark that all our robustness results are stated for the case of a systematic overestimation of the true volatility and jump sensitivity, but they immediately generalize to an underestimation and lead to corresponding subhedges. 

\ \\

\textbf{Notation.} Denote by $\mathbb{R}_{+}$ the positive real numbers. For any $T > 0$ and a function $f:[0,T] \times \mathbb{R} \to \mathbb{R}$, we write $f \in C^{i,j}$ if $f$ is $i$ (resp. $j$) times continuously differentiable w.r.t. the $i$th (resp. $j$th) variable. Further, for $t \in (0,T)$ and $f: [0,T] \times \mathbb{R} \times \mathbb{R} \to \mathbb{R}$ such that $f \in C^{1,2,\cdot}$, we define $\dot{f}(t,\cdot,\cdot) := \frac{\partial}{\partial t}f(t,\cdot,\cdot)$ and for any $s \in \mathbb{R}$ we denote $f'(\cdot,s,\cdot) := \frac{\partial}{\partial s}f(\cdot,s,\cdot).$ For any function $f: \mathbb{R} \rightarrow \mathbb{R}$, we may write $f'_{+}(x) = \lim_{y \downarrow x} \frac{f(x)-f(y)}{x-y}$ for the right-hand derivative and  $f'_{-}(x) = \lim_{y \uparrow x} \frac{f(x)-f(y)}{x-y}$ for the left-hand derivative of $f$ at $x$ provided they exist. Denote the Borel $\sigma$-algebra of the set $\mathcal{X}$ by $\mathcal{B}(\mathcal{X}).$ On some probability space $(\Omega, \mathcal{F}, \mathbb{P}),$ equalities and inequalities between random variables are understood to hold $\mathbb{P}$-a.s.; two random variables are identified if they are equal a.s. We write $L^{p}$ for the space of $\mathbb{R}$-valued $\mathcal{F}$-measurable random variables $X$ such that $||X||_{L^{p}}:=\left(\mathbb{E}[|X|^{p}]\right)^{1/p} < \infty$, for $p \in [1,\infty).$ The equivalence between any two probability measures $\mathbb{P}$ and $\mathbb{Q}$ is denoted by $\mathbb{P} \sim \mathbb{Q}.$  For some semimartingale $Y$, we write $\mathcal{E}(Y)$ for its stochastic exponential. The minimum of two numbers $a,b \in \mathbb{R}$ is denoted by $a \wedge b.$ We write $sgn$ for the sign function.

\section{Model Setup} 
Consider a continuous-time setting with a horizon $T>0$. Let $(\Omega, \mathcal{F}, \mathbb{P})$ be a probability space that is equipped with a standard one-dimensional Brownian motion $W=(W(t))_{t \in [0,T]}$ and a Poisson random measure $J(dt,dz)$ on $[0,T] \times \I,$ being independent of $W,$ with respective intensity measure $\vartheta(dz)dt.$ Denote its compensated version by $\tilde{J}(dt,dz) = J(dt,dz) - \vartheta(dz)dt.$ Let $(\mathcal{F}_{t})_{t \in [0,T]}$ be the right-continuous completion of the filtration generated by $W$ and $J$. \par Throughout the paper we assume that two assets are continuously traded in a frictionless financial market. One of them is the money market whose price at any time $t \in [0,T]$ we denote by $M(t)$ with
\begin{equation}
M(t) = e^{\int_{0}^{t} r(u) du},
\label{eq:moneymarket}
\end{equation} 
for some deterministic interest rate process $r$ satisfying $\int_{0}^{T} |r(u)|\ du < \infty.$ The other asset is denoted by $S$ and henceforth regarded as the \emph{true} stock price process whose realization is displayed at the stock exchange. Its price process $S=(S(t))_{t \in [0,T]}$ satisfies the SDE
\begin{equation}
\frac{dS(t)}{S(t\m)} = r(t) \ dt + \sigma(t) \ dW(t) + \int_{\I} \eta(t,z)\ \tilde{J}(dt,dz),
\label{eq:r_stock}
\end{equation}
\noindent
with $S(0\m) = S(0) > 0.$ The processes $\sigma$ and $\eta$ are assumed to be $(\mathcal{F}_{t})_{t \in [0,T]}$-predictable, $\sigma$ is non-negative and satisfies $\int_{0}^{T} \sigma(t)^{2} \ dt < \infty$ a.s., while $\eta$ is strictly larger than $-1$ and satisfies $\int_{0}^{T} \int_{\I} \eta(t,z)^{2}\ \vartheta(dz)\ dt < \infty$. We denote the jump of $S$ at time $t$ by $\Delta S(t,z) = S(t\m) \eta(t,z).$ In \eqref{eq:r_stock} the mean rate of return is equal to the interest rate $r(t)$, therefore under $\mathbb{P}$ the discounted version of $S$ is a local martingale. As we are mainly interested in the calculation of prices of contingent claims in this work, we restrict to risk-neutral modeling and do not consider the statistical probability measure of $S$. We note, however, that under the physical measure only the drift in \eqref{eq:r_stock} and \eqref{eq:m_stock} below would change which does not have any consequences on hedging and pricing regardless of possible misspecification. The measure $\mathbb{P}$, under which $S$ in \eqref{eq:r_stock} is specified, can be thought of as reference risk-neutral measure. 
We impose the following assumption:

\begin{assumption}
We assume that the local martingale $\tilde{S} := S/M$ is a square-integrable martingale, that is, $\tilde{S} = (\tilde{S}(t))_{t \in [0,T]}$ is a martingale and it holds that
$\mathbb{E}[\tilde{S}(t)^{2}] < \infty, \ \ t \in [0,T].$
\label{ass_martingale}
\end{assumption} 
\noindent
In order to consider options written on $S$, we define payoff functions as follows:

\begin{definition}
A \emph{payoff function} is a convex function $h: \mathbb{R}_{+} \to \mathbb{R}$ having bounded one-sided derivatives, that is
$
|h'_{\pm}(x)| \leq L,\ x > 0,
$
for some positive constant $L$.
\label{def_payoff}
\end{definition}
\noindent
In the sequel, $h$ refers to a non-further specified arbitrary payoff function. A \emph{European contingent claim} is a non-path-dependent contract paying $h(S(T))$ at time $T.$ As the financial market defined by \eqref{eq:moneymarket} - \eqref{eq:r_stock} is generically incomplete, $h(S(T))$ is not necessarily perfectly replicable. This gives rise to the study of possible hedging strategies. We call a bounded predictable process $y=(y(t))_{t \in [0,T]}$ a \emph{self-financing trading strategy}, and the induced portfolio process $P^{y} = (P^{y}(t))_{t \in [0,T]}$ described by the SDE

\begin{equation}
dP^{y}(t) = P^{y}(t) r(t)\ dt + y(t) [dS(t) - r(t)S(t)\ dt], \ P^{y}(0)>0,
\label{eq:portfolio}
\end{equation}
\noindent
whose solution is actually given by

\begin{equation}
P^{y}(t) = M(t) \left[P^{y}(0) + \int_{0}^{t} y(u\m)\ d\tilde{S}(u) \right],\ \ t \in [0,T],
\label{eq:sup_portfolio}
\end{equation}
\noindent
is the \emph{hedging portfolio}. Changes in the value of the portfolio process defined by \eqref{eq:portfolio} are caused only by movements in the assets' price processes and trading gains. In particular, no money is inserted or withdrawn. Due to Assumption \ref{ass_martingale} the process $P^{y}/M$ is a martingale. 

Suppose an investor seeking to hedge $h(S(T))$ knows that the dynamics of $S$ is driven by a Brownian motion supplemented by jumps, but in her pricing and hedging model incorrectly specifies the volatility process and the jump sensitivity. Define the misspecified stock price process $S_{m}^{x} = (S_{m}^{x}(t))_{t \in [0,T]}$ as solution of the SDE

\begin{equation}
\frac{dS_{m}^{x}(t)}{S_{m}^{x}(t\m)} = r(t) \ dt + \gamma(t,S_{m}^{x}(t))\ dW(t) + \int_{\I} \tilde{\gamma}(t,S_{m}^{x}(t\m),z) \ \tilde{J}(dt,dz).
\label{eq:m_stock}
\end{equation}
The dependence on the initial price $S_{m}^{x}(0) = x > 0$ is expressed by the superscript $x.$ When referring to the realized price at some time $t > 0$, we are going to add it to the superscript. Denote the jump of $S_{m}^{x}$ at time $t$ by $\Delta S_{m}^{x}(t,z) = S_{m}^{x}(t\m) \tilde{\gamma}(t,S_{m}^{x}(t\m),z).$  The functions $\gamma$ and $\tilde{\gamma}$ are respectively the misspecified volatility and jump sensitivity. These functions are presumed to be random only through their dependence on the stock price $S_{m}^{x}.$ The subscript $m$ indicates the stock price with the \emph{misspecified} volatility and jump sensitivity. This paper  investigates the impact of a systematic overestimation of the latter on Delta hedging strategies of $h(S(T)).$ To be more precise, we assume that

\begin{equation}
\sigma(t) \leq \gamma(t,S(t)) \ \text{and}\ \text{sgn}(\tilde{\gamma}(t,\tilde{S}_{m}^{x}(t),z) - \eta(t,z)) = \text{sgn}(\eta(t,z)),
\label{eq:cond_mon}
\end{equation}
$d\mathbb{P} \times dt$ and $d\mathbb{P} \times dt \times \vartheta(dz)$-a.s. The second part of the previous condition obviously means that a positive jump sensitivity is always systematically overestimated and a negative one underestimated, i.e., 
\begin{align*}
&\tilde{\gamma}(t,S(t),z) \geq \eta(t,z), \text{if}\ \eta(t,z) \geq 0, \\ & \tilde{\gamma}(t,S(t),z) \leq \eta(t,z), \text{if}\ \eta(t,z) < 0.
\end{align*}
\noindent
Weaker assumptions are possible but are not considered for the ease of exposition. Under \eqref{eq:cond_mon} robustness properties are established in this paper. In order to enable this we need the following assumption:
\begin{assumption} Consider the process $S_{m}^{x}$ defined by \eqref{eq:m_stock}.
\begin{enumerate}[(i)]
	\item Assume that $\gamma:[0,T] \times \mathbb{R}_{+} \rightarrow \mathbb{R}$ is continuous and bounded from above. Defining $\rho(t,s) := s\gamma(t,s),$ suppose that $\rho'(t,s)$ is continuous in $(t,s)$ and locally Lipschitz continuous and bounded in $s \in \mathbb{R}_{+}$, uniformly in $t \in [0,T].$
	\item Assume that $\tilde{\gamma}:[0,T] \times \mathbb{R}_{+} \times \I \rightarrow \mathbb{R}$ is continuous, bounded from above and $\tilde{\gamma}(t,s,z) > -1.$ Defining $\tilde{\rho}(t,s,z) := s\tilde{\gamma}(t,s,z),$ suppose that $\tilde{\rho}'(t,s,z)$ is continuous in $(t,s,z)$, locally Lipschitz continuous and bounded in $s \in \mathbb{R}_{+}$, uniformly in $t \in [0,T]$, $\tilde{\rho}'(t,s,z) > -1+\epsilon$, for some $ \epsilon > 0,$ and that there exists a constant $L > 0$ such that
	\begin{align*}
	&\int_{\I}(\tilde{\rho}(t,s_{1},z) - \tilde{\rho}(t,s_{2},z))^{2}\ \vartheta(dz) \leq L\cdot |s_{1}-s_{2}|^{2},\\
	&\int_{\I} \tilde{\rho}'(t,s,z)^{2}\ \vartheta(dz) \leq L. 
\end{align*}	
\end{enumerate}
\label{ass_vol}
\end{assumption} 
\noindent
Assumption \ref{ass_vol}(i) is also used in \cite{kjs} and is a standard one. Assumption \ref{ass_vol}(ii) is novel; we remark that boundedness and Lipschitz conditions are not unusual in the SDE literature with jumps. The condition that $\tilde{\rho}'(t,s,z) > -1+\epsilon$ is needed to derive the monotonicity of the stochastic flows in the initial value, which is needed to show convexity of the contingent claim's price. A counterexample is provided below showing that if $\tilde{\rho}'(t,s,z) \leq -1,$ the effect of a jump could be larger than the value of the stock before the jump and would be powerful enough to flip the sign, and thus reverse the ordering. This would lead to a stock price process which is not monotone in the initial value.
\par Denote the discounted version of $S_{m}^{x}$ by $\tilde{S}_{m}^{x} := S_{m}^{x}/M.$ Define

$$\mathcal{Q}_{em} := \{\Q \sim \mathbb{P}| \tilde{S}_{m}^{x}\ \text{is a martingale w.r.t.}\ \Q\},$$
i.e., $\mathcal{Q}_{em}$ is the set of all equivalent martingale measures (EMMs) of $\tilde{S}_{m}^{x}$ (see Lemma \ref{lemma_girsanov} for a detailed characterization). For some subset $\mathcal{M} \subseteq \mathcal{Q}_{em}$, we call a stochastic process $X$ an $\mathcal{M}$-(sub/super-)martingale if it is a (sub/super-)martingale w.r.t. all measures $\mathbb{Q} \in \mathcal{M}.$
%

\section{Convexity of European contingent claim value}
The convexity of the European contingent claim value is the main tool in the proofs in subsequent sections. We formally define it under the reference measure $\mathbb{P}$ as follows:
\begin{definition}
The \emph{misspecified value} at time $t$ of the European contingent claim with payoff function $h$ is 
\begin{equation*}
v_{m}(t,x) := \mathbb{E} \left[h(S_{m}^{t,x}(T))\ e^{-\int_{t}^{T} r(u) du} \right], \ \ t \in [0,T],\ x > 0.
\end{equation*}
\label{def_price}
\end{definition}
\noindent
The next theorem is formulated for the time-zero misspecified price only. The generalization to arbitrary $(t,s) \in [0,T] \times \mathbb{R}_{+}$ is immediate.

\begin{theorem}
Consider the process $S_{m}^{x}$ described by \eqref{eq:m_stock}, suppose that Assumption \ref{ass_martingale} and Assumption \ref{ass_vol} are satisfied. Then the European contingent claim value $v_{m}(x):=v_{m}(0,x)$ is convex in $x$. 
\label{thm_convex}
\end{theorem}
\begin{proof}The proof is conducted in six steps: \\

\noindent
\textsc{Step 1} is to prove: if $0 < x < y$, then  $S_{m}^{x}(T) \leq S_{m}^{y}(T).$\\
 Define $\tau_{1} := \inf\{t \geq 0: S_{m}^{x}(t) > S_{m}^{y}(t)\} \wedge T$ and by contradiction suppose $\tau_{1} \wedge T <T.$ Consequently, it holds that $S^{x}_{m}(\tau_{1}\text{-}) < S^{y}_{m}(\tau_{1}\text{-})$ and in addition
	\begin{align*}
	&S^{x}_{m}(\tau_{1}\text{-}) + \Delta S_{m}^{x}(\tau_{1}\m,z) > S^{y}_{m}(\tau_{1}\text{-}) + \Delta S_{m}^{y}(\tau_{1}\m,z) \\
	\Rightarrow\ &S^{x}_{m}(\tau_{1}\text{-}) - S^{y}_{m}(\tau_{1}\text{-}) > \Delta S_{m}^{y}(\tau_{1}\m,z) - \Delta S_{m}^{x}(\tau_{1}\m,z) =   \tilde{\rho}(\tau_{1},S^{y}_{m}(\tau_{1}\text{-}),z) - \tilde{\rho}(\tau_{1},S^{x}_{m}(\tau_{1}\text{-}),z). 
	\end{align*}
	Observe that
	\begin{align*}
	&\tilde{\rho}(\tau_{1},S^{y}_{m}(\tau_{1}\text{-}),z) - \tilde{\rho}(\tau_{1},S^{x}_{m}(\tau_{1}\text{-}),z) = \int_{S^{x}_{m}(\tau_{1}\text{-})}^{S^{y}_{m}(\tau_{1}\text{-})} \tilde{\rho}'(\tau_{1},s,z)\ ds >   (-1+\epsilon) (S^{y}_{m}(\tau_{1}\text{-}) - S^{x}_{m}(\tau_{1}\text{-})),
	\end{align*}
	so in total we obtain 
	\begin{align*}
	 0 > \epsilon\ (S^{y}_{m}(\tau_{1}\text{-}) - S^{x}_{m}(\tau_{1}\text{-})),
	\end{align*}
	which is obviously a contradiction. Thus, we conclude that $\tau_{1} \wedge T = T.$ This yields that $\tau_{2} := \inf\{t \geq 0: S_{m}^{x}(t) \geq S_{m}^{y}(t)\} \wedge T = \inf\{t \geq 0: S_{m}^{x}(t) = S_{m}^{y}(t)\} \wedge T.$ If $\tau_{2} \wedge T = \tau_{2},$ then strong uniqueness for \eqref{eq:m_stock} (cf. \cite{oksjump}, Theorem 1.19) implies that $S^{x}_{m}(t) = S^{y}_{m}(t)$ for all $t \in [\tau_{2},T].$ To sum up, we see that $S_{m}^{x}(T) \leq S_{m}^{y}(T).$\\
	
\noindent
\textsc{Step 2} is to note that if $0<x<y,$ then by convexity of $h$ and \textsc{Step 1} it holds that
\begin{align}
&h(S_{m}^{y}(T)) - h(S_{m}^{x}(T)) \leq h'_{+}(S_{m}^{y}(T)) \cdot (S_{m}^{y}(T)-S_{m}^{x}(T)),\ \label{eq:bound1}  \\
&h(S_{m}^{y}(T)) - h(S_{m}^{x}(T)) \geq h'_{+}(S_{m}^{x}(T)) \cdot (S_{m}^{y}(T) - S_{m}^{x}(T)). \label{eq:bound2}
\end{align}

\noindent
\textsc{Step 3} is to prove: if $x,y> 0, x\neq y,$ then $\phi(t) := \mathbb{E}\left[\left(\frac{S_{m}^{y}(t)-S_{m}^{x}(t)}{M(t)}\right)^{2} \right] \leq 3(y-x)^{2} e^{6L^{2}T}, \ t\in [0,T]$ and $L > 0.$\\
We find that
\begin{align*}
\phi(t) &
\leq3(y-x)^{2} + 3 \mathbb{E} \left(\int_{0}^{t} \left(\frac{\rho(u,S_{m}^{y}(u))-\rho(u,S_{m}^{x}(u))}{M(u)}\right)^{2} du \right) \\
&\ + 3 \mathbb{E} \left(\int_{0}^{t} \int_{\I} \left(\frac{\tilde{\rho}(u,S_{m}^{y}(u),z)-\tilde{\rho}(u,S_{m}^{x}(u),z)}{M(u)}\right)^{2} \vartheta(dz)du \right) \\
&\leq 3(y-x)^{2} + 6L^{2} \int_{0}^{t} \phi(u) \ du,
\end{align*}
whereby the first inequality follows from an elementary inequality and \emph{It\^{o}'s isometry}  while the second is justified by \textit{Lipschitz continuity} and \emph{Tonelli's Theorem}. Next, \emph{Gr\"{o}nwall's inequality} implies
\begin{equation}
\phi(t) \leq 3(y-x)^{2} e^{6L^{2}t}, \ t \in [0,T],
\label{eq:UI}
\end{equation}
and observing that the right-hand side of \eqref{eq:UI} is increasing in $t$ gives the claim.\\

\noindent
\textsc{Step 4} is to prove: if $x,y > 0, x\neq y$, then 
\begin{align}
v'_{m,+}(x) &= \lim_{y \downarrow x} \frac{v_{m}(y) - v_{m}(x)}{y-x} = \mathbb{E}\left[h'_{+}(S_{m}^{x}(T)) \xi_{m}^{x}(T) \right],\  \label{eq:right_der} \\
v'_{m,-}(x) &= \lim_{y \uparrow x} \frac{v_{m}(y) - v_{m}(x)}{y-x} = \mathbb{E}\left[h'_{-}(S_{m}^{x}(T)) \xi_{m}^{x}(T) \right],
\label{eq:left_der}
\end{align}
with $\xi_{m}^{x}(T) = \mathcal{E} \left(\int_{0}^{T} \rho'(u,S_{m}^{x}(u\m)) \ dW(u) + \int_{0}^{T} \int_{\I}  \tilde{\rho}'(u,S_{m}^{x}(u\m),z)\ \tilde{J}(du,dz) \right).$\\
Observe that our assumptions allow us to employ the theory of \emph{stochastic flows} for general semimartingales (we use \cite{protter}, Chapter V.7, Theorem $39$ in the sequel; see also \cite{kunita2004stochastic} and the references therein): for almost all $\omega \in \Omega$ the function $x \mapsto S_{m}^{x}(t)$ is continuously differentiable. Phrased differently, there exists $\mathcal{N}_{1}$ with $\mathbb{P}(\mathcal{N}_{1}) =0$ such that for all $\omega \in \Omega \setminus \mathcal{N}_{1}$ the function $D_{m}^{x}(t) := (\partial / \partial x) S_{m}^{x}(t)$ is defined. We only consider such $\omega$ in the sequel. Then $D_{m}^{x}(t)$ solves the SDE given by
$$dD_{m}^{x}(t) = D_{m}^{x}(t\text{-}) \left[r(t)\ dt + \rho'(t,S_{m}^{x}(t)) \ dW(t) + \int_{\I} \tilde{\rho}'(t,S_{m}^{x}(t),z)\ \tilde{J}(dt,dz)\right],$$
with $D_{m}^{x}(0) = 1.$ An application of \emph{It\^{o}'s formula} then yields 
\begin{align}
D_{m}^{x}(t) &= \xi_{m}^{x}(t)M(t) \label{eq:flow1}, \\
\xi_{m}^{x}(t) &= \mathcal{E} \left(\int_{0}^{t} \rho'(u,S_{m}^{x}(u\m)) \ dW(u) + \int_{0}^{t} \int_{\I}  \tilde{\rho}'(u,S_{m}^{x}(u\m),z)\ \tilde{J}(du,dz) \right) \label{eq:xi}.
\end{align}
We only argue for \eqref{eq:right_der} since \eqref{eq:left_der} is established analogously. Observe that
\begin{align*}
\limsup_{y \downarrow x} \frac{v_{m}(y) - v_{m}(x)}{y-x} & \stackrel{\eqref{eq:bound1}}{\leq} \limsup_{y \downarrow x} \mathbb{E} \left[h'_{+}(S_{m}^{y}(T)) \frac{S_{m}^{y}(T) - S_{m}^{x}(T)}{M(T)(y-x)} \right]\\
& =  \mathbb{E} \left[\limsup_{y \downarrow x} h'_{+}(S_{m}^{y}(T)) \frac{S_{m}^{y}(T) - S_{m}^{x}(T)}{M(T)(y-x)} \right]= \mathbb{E}[h'_{+}(S_{m}^{x}(T)) \xi_{m}^{x}(T)],
\end{align*}
whereby the \emph{uniform integrability} used in the first equality is implied by \textsc{Step 3} and the last equality holds since $h$ has bounded one-sided right-continuous derivatives and because of \eqref{eq:flow1}. Conversely, it follows similarly from \eqref{eq:bound2} that

$$\liminf_{y \downarrow x} \frac{v_{m}(y) - v_{m}(x)}{y-x} = \mathbb{E}[h'_{+}(S_{m}^{x}(T)) \xi_{m}^{x}(T)].$$

\noindent
We conclude that \eqref{eq:right_der} holds.\\

\noindent
\textsc{Step 5} is to prove: $\xi_{m}^{x} = (\xi_{m}^{x}(t))_{t \in [0,T]}$ given by \eqref{eq:xi} is a positive martingale. \\
Since $\rho'$ and $\tilde{\rho}'$ are bounded in $s$ (uniformly in $t$ respectively in $(t,z)$), the process $$\left(\int_{0}^{t} \rho'(u,S_{m}^{x}(u\m)) \ dW(u) + \int_{0}^{t} \int_{\I}  \tilde{\rho}'(u,S_{m}^{x}(u\m),z)\ \tilde{J}(du,dz)\right)_{t \in [0,T]}$$ is a martingale of bounded mean oscillation under $\mathbb{P},$ also referred to as BMO($\mathbb{P}$)-martingale\footnote{For a martingale $X$ with c\`adl\`ag paths, denote the quadratic variation by $[X,X]$ and define the $\mathcal{H}^{p}$ norm of $X$ for any $p\geq 1$ by $||X||_{\mathcal{H}^{p}} := \mathbb{E}\left[[X,X]_{T}^{p/2}\right]^{1/p}.$ A martingale $X$ is said to be of \emph{BMO} if it is in $\mathcal{H}^{2}$ and there exists a constant $c$ such that for any stopping time $\tau \leq T$ it holds that $\mathbb{E}\left[(X_{T}-X_{\tau\m})^{2} | \mathcal{F}_{\tau}\right] \leq c^{2}$ a.s. (see \cite{protter}, Chapter IV.4 for further details).}. Moreover, as $\tilde{\rho}' > -1+\epsilon$,  \emph{Kazamaki's criterion} (see  \cite{kazamaki}) yields the claim. \\

\noindent
\textsc{Step 6} is to prove: $v_{m,\pm}'$ is non-decreasing.\\
Define a new probability measure $\mathbb{P}^{x}$ on $(\Omega, \mathcal{F})$ by $d\mathbb{P}^{x} / d\mathbb{P} = \xi_{m}^{x}(T).$ According to \textsc{Step 5}, we can apply \emph{Girsanov's theorem} (cf. Lemma \ref{lemma_girsanov}) to deduce that
\begin{equation*}
W^{x}(t) = W(t) - \int_{0}^{t} \rho'(u,S_{m}^{x}(u)) \ du
\end{equation*}
is a $\mathbb{P}^{x}$- Brownian motion and 
\begin{equation*}
\tilde{J}^{x}(dt,dz) = \tilde{J}(dt,dz) - \tilde{\rho}'(t,S_{m}^{x}(t),z)\ \vartheta(dz)dt
\end{equation*}
is a $\mathbb{P}^{x}$- compensated Poisson random measure. In particular, the $\mathbb{P}^{x}$-compensator $\vartheta^{x}(dz)dt$ of $J(dt,dz)$ is given by
$$\vartheta^{x}(dz)dt :=\left(\tilde{\rho}'(t,S_{m}^{x}(t),z)+1\right) \vartheta(dz)dt.$$
Consequently, we can rewrite \eqref{eq:m_stock} as 
\begin{align*}
dS_{m}^{x}(t) &= S_{m}^{x}(t) r(t)\ dt + \rho(t,S_{m}^{x}(t)) \rho'(t,S_{m}^{x}(t))\ dt + \rho(t,S_{m}^{x}(t))\ dW^{x}(t) \\
&\ + \int_{\I} \tilde{\rho}(t,S_{m}^{x}(t\m),z)\ \tilde{J}^{x}(dt,dz) + \int_{\I} \tilde{\rho}(t,S_{m}^{x}(t\m),z)\ \tilde{\rho}'(t,S_{m}^{x}(t\m),z)\ \vartheta(dz)dt,
\end{align*}
with $S_{m}^{x}(0) = x,$ and note that \emph{uniqueness in law} holds for solutions to this SDE (\cite{applebaum}, Chapter 6). Define $\bar{S}_{m}^{x}$ as solution to
\begin{align*}
d\bar{S}_{m}^{x}(t) &= \bar{S}_{m}^{x}(t) r(t)\ dt + \rho(t,\bar{S}_{m}^{x}(t)) \rho'(t,\bar{S}_{m}^{x}(t))\ dt + \rho(t,\bar{S}_{m}^{x}(t))\ dW(t) \\
&\ + \int_{\I} \tilde{\rho}(t,\bar{S}_{m}^{x}(t\m),z)\ \tilde{J}(dt,dz)  + \int_{\I} \tilde{\rho}(t,\bar{S}_{m}^{x}(t\m),z)\ \tilde{\rho}'(t,\bar{S}_{m}^{x}(t\m),z)\ \vartheta(dz)dt,
\end{align*}
with $\bar{S}^{x}_{m}(0) = x.$ Observe that the process $\bar{S}^{x}_{m}$ has the same distribution under $\mathbb{P}$ as the process $S^{x}_{m}$ under $\mathbb{P}^{x}.$ Calculating $v'_{m,\pm}$ under $\mathbb{P}^{x}$, we conclude from \textsc{Step 4} that $v'_{m,+}(x) = \mathbb{E}^{x}[h'_{+}(S_{m}^{x}(T))]$ and $v'_{m,-}(x) = \mathbb{E}^{x}[h'_{-}(S_{m}^{x}(T))]$  and therefore
\begin{equation}
v'_{m,\pm}(x) = \mathbb{E}[h'_{\pm}(\bar{S}^{x}_{m}(T))], \ x > 0.
\label{eq:der_v}
\end{equation}
For $0<x<y$, following the same line of reasoning as in \textsc{Step 1}, one can show that $\bar{S}^{x}_{m}(T) \leq \bar{S}^{y}_{m}(T).$ Since $h'_{\pm}$ is non-decreasing, monotonicity of the expected value implies that $v'_{m,\pm}$ is also non-decreasing, which is equivalent to $v_{m}(x)$ being convex w.r.t. $x$. 
\end{proof}
Note that Theorem \ref{thm_convex} generalizes Theorem 5.2 in \cite{kjs} to jump-diffusions. We conclude from Theorem \ref{thm_convex} that the Delta strategy for the misspecified model always exists and is bounded. The following example shows that the condition $\tilde{\rho}'(t,s,z) > -1+ \epsilon,\ \epsilon > 0,$ enforced in the second part of Assumption \ref{ass_vol} is not only necessary in \textsc{Step 5} of the proof of Theorem \ref{thm_convex}, but that it is also inevitable to establish the monotonicity of the mapping $x \mapsto S_{m}^{x}(t)$ in \textsc{Step 1}, without which convexity would not hold. 

\begin{example} 
Suppose $r \equiv 0, \tilde{\rho}'(t,s,z) = -2$  and the stock price is driven by a compensated homogeneous Poisson process $\tilde{N} = (\tilde{N}(t))_{t \in [0,T]}$ with intensity $\lambda > 0.$ Assuming the constant of integration is equal to $4$, the stock price dynamics is given by
$$dS_{m}^{x}(t) = (-2S_{m}^{x}(t\m)+4) \ d\tilde{N}(t).$$
\noindent
Denote the first jump time of the Poisson process by $\tau := \inf\{t > 0: N(t) =1\}.$ Then we obviously have
\begin{align*}
S_{m}^{x}(\tau) = x + (-2x+4)\cdot (1-\lambda \tau).
\end{align*}
Choosing $\lambda = 0.1$ and defining $B := \{\omega \in \Omega: \tau(\omega) <5\},$ we see that
$$S_{m}^{1}(\tau) \mathbbm{1}_{B} = (3-0.2 \tau) \mathbbm{1}_{B} > 2 = S_{m}^{2}(\tau)\mathbbm{1}_{B},$$
\noindent
i.e., the monotonicity property no longer holds a.s. because $\mathbb{P}(B) > 0.$ Obviously, the parameters can be chosen such that $B$ has probability arbitrarily close to one.
\end{example}

\section{Robustness of the Delta hedge}

In this section we consider an investor intending to approximate $h(S(T))$ by means of its Delta strategy. It is assumed that the investor bases her computation of the Delta strategy on the misspecified model price \eqref{eq:m_stock} under the reference measure $\mathbb{P}$. We first characterize the induced hedging error and subsequently deduce from its characteristics certain robustness properties. Throughout this and the next section we need the following:

\begin{definition}
We say a measure $\mathbb{Q} \in \mathcal{Q}_{em}$ satisfies \emph{Condition (I)} if it holds for every  $(t,s) \in [0,T] \times \mathbb{R}_{+}$ that there exists some constant $L>0$ so that
\begin{align*}
&\int_{\I} \tilde{\rho}(t,s,z)^{2}\ \vartheta_{\mathbb{Q}}(dz) \leq L\cdot(1+|s|^{2}), \quad \text{and}\quad
&S_{m}^{x} \in L^{2}(d\mathbb{Q} \times dt), \ x >0.
\end{align*}
We denote $\mathcal{Q} := \{\mathbb{Q} \in \mathcal{Q}_{em}|\ \mathbb{Q} \ \text{satisfies \emph{Condition (I)}}\}.$
\label{def_mart_Q}
\end{definition}
\noindent

\noindent

Consider for all $x \in \mathbb{R}_{+}$ and $t \in [0,T]$ the partial integro-differential equation (PIDE) given by
\begin{align}
\begin{split}
0 &= \dot{g}(t,x) + r(t)x g'(t,x) + \frac{1}{2} x^{2} \gamma^{2}(t,x)g''(t,x) - r(t) g(t,x)\\
 &\ \ + \int_{\I} \left(g(t,x+x\tilde{\gamma}(t,x,z)) - g(t,x) - x\tilde{\gamma}(t,x,z) g'(t,x) \right) \vartheta(dz),\\
g(x,T) &= h(x).
\end{split}
\label{eq:FK}
\end{align}
\noindent
The following assumption is an extension of the hypotheses needed in Section 6 of \cite{kjs}.

\begin{assumption}
We assume the existence of a classical solution $g:[0,T] \times \mathbb{R}_{+} \to \mathbb{R}$ to the PIDE \eqref{eq:FK} whose derivatives in the second variable are bounded by a polynomial function of $x$, uniformly in $t \in [0,T].$
\label{ass_regularity_V}
\end{assumption}
\noindent
For results on the existence of classical solutions to \eqref{eq:FK} see for instance \cite{bensoussan1982controle} or \cite{contv} (for the purely Brownian case conditions are discussed in \cite{friedman}). The \emph{Feynman-Kac theorem} (cf. \cite{fk} and the references therein) then implies that the solution $g$ to the PIDE \eqref{eq:FK} coincides with the misspecified value of the contingent claim from Definition \ref{def_price}, i.e., 
\begin{equation}
g(t,x) = v_{m}(t,x) = \mathbb{E} \left[h(S_{m}^{t,x}(T))\ e^{-\int_{t}^{T} r(u)\ du} \right],\ t \in [0,T],\ x>0.
\label{eq:mis_price}
\end{equation}
We immediately obtain the following corollary that is necessary to extract trading strategies and uniform bounds from price functions:

\begin{corollary}
Suppose the conditions of Theorem \ref{thm_convex} are satisfied. Then the convex European contingent claim value function $v_{m}$ given by \eqref{eq:mis_price} satisfies
$$|v_{m,\pm}'(x)| \leq ||h_{\pm}'||_{L^{\infty}} = \sup_{y \in \mathbb{R}_{+}} |h_{\pm}'(y)|, \ x > 0.$$
\label{cor_bounded}
\end{corollary}
\begin{proof}
The claim is immediate from Definition \ref{def_payoff} and \eqref{eq:der_v}. 
\end{proof}

Note that $v_{m}'$ is an admissible trading strategy. Suppose now the investor follows the Delta strategy $v_{m}' = (v_{m}'(t,\cdot))_{t \in [0,T]}$.  The trading is done in the physical stock $S$, whose price is described by \eqref{eq:r_stock}. Then the corresponding self-financing hedging portfolio $P^{v_{m}'} = (P^{v_{m}'}(t))_{t \in [0,T]}$ solves the SDE

\begin{equation}
dP^{v_{m}'}(t) = P^{v_{m}'}(t)r(t) \ dt + v_{m}'(t,S(t\m)) [dS(t) - r(t)S(t)\ dt], 
\label{eq:hedging_portfolio}
\end{equation}
and the initial capital to set up the hedging portfolio coincides with the initial misspecified price of the claim, i.e., $P^{v_{m}'}(0) = v_{m}(0,x).$ Observe that $P^{v_{m}'}$ is typically not Markov. We formally define the hedging error (in the misspecified model) $e_{m} = (e_{m}(t))_{t \in [0,T]}$ by
\begin{equation}
e_{m}(t) := P^{v_{m}'}(t) - v_{m}(t,S(t)).
\label{eq:hedging_error}
\end{equation}
This hedging error at time $t$ obviously displays the difference between the value of the hedging portfolio and the misspecified claim price upon observing $S(t)$ quoted at the stock exchange. We are particularly interested in the hedging error's value at time $T$ because $v_{m}(T,S(T)) = h(S(T)).$ So $e(T)$ indicates how far off an investor following $P^{v_{m}'}$ is from the claim's payoff due to both misspecification and market incompleteness. The next proposition gives an explicit formula for the discounted hedging error $e_{m}/M = (e_{m}(t)/M(t))_{t \in [0,T]}.$

\begin{proposition}
Suppose Assumption \ref{ass_martingale}, Assumption \ref{ass_vol} and Assumption \ref{ass_regularity_V} are satisfied. Consider the hedging error $e_{m}$ defined by \eqref{eq:hedging_error}. The \textit{discounted hedging error} reads
\begin{align}
\begin{split}
\frac{e_{m}(t)}{M(t)} &= \frac{1}{2} \int_{0}^{t} \frac{1}{M(u)}v''_{m}(u,S(u)) S(u)^{2}[\gamma(u,S(u))^{2}- \sigma(u)^{2}]\ du \\
&\ \ + \int_{0}^{t}\int_{\I} \frac{1}{M(u)}\Big(v_{m}(u,S(u\m)+ S(u\m)\tilde{\gamma}(u,S(u\m),z)) - v_{m}(u,S(u\m) + S(u\m) \eta(u,z))  \\
&\ \ \phantom{+ \int_{0}^{t}\int_{\I} \frac{1}{M(u)}\Big(} + v'_{m}(u,S(u))S(u\m)(\eta(u,z)-\tilde{\gamma}(u,S(u\m),z))\Big)\ \vartheta(dz)du \\
&\ \ - \int_{0}^{t} \int_{\I} \frac{1}{M(u)} \Big(v_{m}(u,S(u\m) + S(u\m)\eta(u,z)) - v_{m}(u,S(u\m)) \\
&\ \ \phantom{- \int_{0}^{t} \int_{\I} \frac{1}{M(u)} \Big(} - S(u\m)\eta(u,z) v'_{m}(u,S(u)) \Big)\ \tilde{J}(du,dz).
\end{split}
\label{eq:disc_error}
\end{align}
\label{hedging_error}
\end{proposition}
\begin{proof}
Observe that
\begin{align*}
&dv_{m}(t,S(t)) \notag \\
&= \dot{v}_{m}(t,S(t)) \ dt + v_{m}'(t,S(t))\ dS(t) + \frac{1}{2} v_{m}''(t,S(t)) S(t)^{2} \sigma(t)^{2} \ dt \notag \\
&\ + \int_{\I}\left( v_{m}(t,S(t\m) + S(t\m) \eta(t,z)) - v_{m}(t,S(t\m)) - v_{m}'(t,S(t\m)) S(t\m) \eta(t,z)\right) \ J(dt,dz) \notag \\
\begin{split}
&= r(t) v_{m}(t,S(t)) \ dt + v'_{m}(t,S(t)) [dS(t) - r(t)S(t) \ dt] \\
&\ + \frac{1}{2} v''_{m}(t,S(t)) S(t)^{2}[\sigma(t)^{2} - \gamma(t,S(t))^{2}]\ dt  \\
&\ + \int_{\I} \Big( v_{m}(t,S(t\m) + S(t\m)\eta(t,z)) - v_{m}(t,S(t\m)) - S(t\m)\eta(t,z) v'_{m}(t,S(t)) \Big) \tilde{J}(dt,dz)  \\
&\ + \int_{\I} \Big(v_{m}(t,S(t\m) + S(t\m) \eta(t,z)) - v_{m}(t,S(t\m)+ S(t\m)\tilde{\gamma}(t,S(t\m),z))  \\
&\ \phantom{+ \int_{\I} \Big(} + v'_{m}(t,S(t))S(t\m)(\tilde{\gamma}(t,S(t\m),z) - \eta(t,z)) \Big)\ \vartheta(dz)dt,
\label{eq:step_before}
\end{split}
\end{align*}
where the first equality is an application of \emph{It\^{o}'s lemma} and the second one follows from \eqref{eq:FK}. Next we calculate the difference $de_{m}(t) = d(P^{v_{m}'}(t) - v_{m}(t,S(t))).$  This yields
\begingroup
\allowdisplaybreaks 
\begin{align*}
de_{m}(t) &= (\underbrace{P^{v_{m}'}(t) - v_{m}(t,S(t))}_{= e_{m}(t)})\ r(t)\ dt + \frac{1}{2} v_{m}''(t,S(t))S(t)^{2}[\gamma(t,S(t))^{2}- \sigma(t)^{2}]\ dt \\
&\ \ + \int_{\I} \big(v_{m}(t,S(t\m) + S(t\m) \tilde{\gamma}(t,S(t\m),z)) - v_{m}(t,S(t\m) + S(t\m) \eta(t,z))\\
&\ \ \phantom{+ \int_{\I} \big(} - v_{m}'(t,S(t\m))S(t\m) (\tilde{\gamma}(t,S(t\m),z) - \eta(t,z)) \big)\ \vartheta(dz)dt \\
&\ \ - \int_{\I} \big(v_{m}(t,S(t\m)+S(t\m) \eta(t,z)) - v_{m}(t,S(t\m)) - S(t\m) \eta(t,z) v_{m}'(t,S(t\m)) \big) \ \tilde{J}(dt,dz).
\end{align*}
\endgroup
An application of the integration by parts formula to $d(e_{m}(t)/M(t))$ then establishes \eqref{eq:disc_error}.
\end{proof}

The following lemma provides among others valuable insight into monotonicity properties of the (expected) discounted hedging error. 

\begin{lemma}
Consider the process $S$ given by \eqref{eq:r_stock} and let $g$ be some function that is convex in the second component. Assume further that
\begin{equation}
\sigma(t) \leq \gamma(t,S(t)) \ \text{and} \ \emph{sgn}(\tilde{\gamma}(t,S(t),z)-\eta(t,z)) = \emph{sgn}(\eta(t,z)),
\label{eq:cond_vol}
\end{equation}
$d\mathbb{P} \times dt$ and $d\mathbb{P} \times dt \times \vartheta(dz)$-a.s. Then the process $\Pi = (\Pi(t))_{t \in [0,T]}$ defined by
\begin{equation}
\begin{split}
\Pi(t) &:= \frac{1}{2} \int_{0}^{t} g''(u,S(u)) S(u)^{2} [\gamma(u,S(u))^{2} - \sigma(u)^{2}] \ du   \\
&\ \ + \int_{0}^{t} \int_{\I} \Big(g(u,S(u\m) + S(u\m) \tilde{\gamma}(u,S(u\m),z)) - g(u,S(u\m) + S(u\m) \eta(u,z)) \\
&\ \ \phantom{+ \int_{0}^{t} \int_{\I} \Big(} + g'(u,S(u\m)) S(u\m) (\eta(u,z) - \tilde{\gamma}(u,S(u\m),z)) \Big)\ \vartheta(dz) du,
\end{split}
\label{eq:pit}
\end{equation}
is non-decreasing a.s.
\label{lem_mon}
\end{lemma}
\begin{proof}
The first integral in \eqref{eq:pit} is easily seen to be non-decreasing because $g''(\cdot,\cdot) > 0$ by convexity of $g$ and the left part of the condition \eqref{eq:cond_vol}. Regarding the integrand of the double integral, observe that
\begin{align*}
& g(t,S(t\m) + S(t\m) \tilde{\gamma}(t,S(t\m),z)) - g(t,S(t\m) + S(t\m) \eta(t,z)) \\
&\ + g'(t,S(t\m)) S(t\m) (\eta(t,z) - \tilde{\gamma}(t,S(t\m),z)) \\
&\geq S(t\m) \left(\tilde{\gamma}(t,S(t\m),z) - \eta(t,z) \right) \cdot (g'(t,S(t\m) + S(t\m) \eta(t,z)) - g'(t,S(t\m))) \geq 0, 
\end{align*}
where the first inequality follows from the convexity of $g$ and the second inequality holds because the right part of the condition \eqref{eq:cond_vol} and the fact that $g'$ is monotonically non-decreasing imply that the two brackets always have the same sign inducing a non-negative product. We conclude that the double integral is non-decreasing as well and so is the process $\Pi.$
\end{proof}

We see that Lemma \ref{lem_mon} implies that the first two summands of the discounted hedging error \eqref{eq:disc_error} are non-decreasing. However, the third term cannot be non-decreasing  because it is a martingale. The fact that the hedging error does not exhibit any a.s. monotonicity properties despite a systematic overestimation has a straightforward economic interpretation: due to the market incompleteness, the Delta strategy does neither provide a perfect hedge nor an a.s. sub- or superhedge. Nonetheless, it is possible to deduce that the expected discounted hedging error is a good-deal with respect to a robustified expectation which can be taken with respect to any reference probability measure. To this end, consider the set of EMMs $\mathcal{Q}$ and its characterization in Lemma \ref{lemma_girsanov}. 
Define 
\begin{equation}
\label{hull}
\mathcal{Q}_{0}(\p) := \{\mathbb{Q} \in \mathcal{Q}|\ \vartheta_{\mathbb{Q}}(A) \leq \vartheta(A)=\vartheta_{\mathbb{P}}(A)\ \text{for any}\ A \in \mathcal{B}(\I)\}.
\end{equation}
\noindent
Note that the set of L\'evy measures satisfying the inequality in the definition of $\mathcal{Q}_{0}(\p)$ is convex. Let $\cl(\conv(\mathcal{Q}_{0}(\p)))$ denote the closure of the convex hull of $\mathcal{Q}_{0}( \p).$  The next theorem states that the discounted hedging error is a $\cl(\conv(\mathcal{Q}_{0}(\p)))$-submartingale meaning that it is a submartingale w.r.t. every measure $\mathbb{Q} \in \cl(\conv(\mathcal{Q}_{0}(\p))).$ 

\begin{theorem}
\label{theoone}
Suppose the conditions of Proposition \ref{hedging_error} are satisfied. If 
\begin{equation*}
\sigma(t) \leq \gamma(t,S(t)) \ \text{and} \ \emph{sgn}(\tilde{\gamma}(t,S(t),z)-\eta(t,z)) = \emph{sgn}(\eta(t,z)),
\end{equation*}
$d\mathbb{P} \times dt$ and $d\mathbb{P} \times dt \times \vartheta(dz)$-a.s., then the induced discounted hedging error $e_{m}/M$ specified by \eqref{eq:disc_error} is a $\cl(\conv(\mathcal{Q}_{0}(\p)))$-submartingale. In particular, $\inf_{\mathbb{Q} \in \cl(\conv(\mathcal{Q}_{0}(\p)))} \mathbb{E}_{\mathbb{Q}}\left[\frac{e_{m}(t)}{M(t)}\right] \geq 0$ for all $t \in [0,T]$ and the inequality is strict if the hedging error is not zero.
\label{thm_superhedge_delta}
\end{theorem}

\begin{proof}\ \\
\textsc{Step 1} is to observe that for every $\mathbb{Q} \in \mathcal{Q}$ and some function $f:\mathbb{R}_{+} \to \mathbb{R}$ such that $f'$ is bounded by some constant $K>0$, it holds that
\begin{align}
\begin{split}
\int_{\I}f'(S(t\m))^{2} S(t)^{2} \eta(t,z)^{2} \ \vartheta_{\mathbb{Q}}(dz) \leq K^{2} \int_{\I} S(t)^{2} \tilde{\gamma}(t,S(t),z)^{2}\ \vartheta_{\mathbb{Q}}(dz)  \leq K^{2}L (1+|S(t)|^{2}).
\end{split}
\label{eq:replace}
\end{align}
\noindent
\textsc{Step 2} is to establish the claim of the theorem.\\
Let $\mathbb{Q} \in \mathcal{Q}_{0}(\p)$ arbitrary. Performing a change of measure from $\mathbb{P}$ to $\mathbb{Q}$ in \eqref{eq:disc_error}, the resulting compensated jump integral is consequently defined under $\mathbb{Q},$ see \eqref{eq:girsanov} for the specification of $\tilde{J}_{\mathbb{Q}}(\cdot,\cdot).$ From Corollary \ref{cor_bounded} and \eqref{eq:replace} we see that the process 
\begin{align*}
\Bigg(\int_{0}^{t}\int_{\I} \big(&v_{m}(u,S(u\m)+S(u\m) \eta(u,z)) - v_{m}(u,S(u\m))\\
& - S(u\m) \eta(u,z) v_{m}'(u,S(u\m)) \big) \ \tilde{J}_{\mathbb{Q}}(du,dz)\Bigg)_{t \in [0,T]}
\end{align*}
is a true martingale under each $\mathbb{Q}.$ Thus, fixing $0 \leq w < t \leq T$, we can calculate that
\begin{align}
\begin{split}
&\mathbb{E}_{\mathbb{Q}}\left[\frac{e_{m}(t)}{M(t)} \Big| \mathcal{F}_{w}\right]\\
 &= \frac{e_{m}(w)}{M(w)} + \mathbb{E}_{\mathbb{Q}}\Bigg[ \frac{1}{2} \int_{w}^{t} \frac{1}{M(u)}v''_{m}(u,S(u)) S(u)^{2}[\gamma(u,S(u))^{2}- \sigma(u)^{2}]\ du \\
&\ + \int_{w}^{t}\int_{\I} \frac{1}{M(u)}\Big(v_{m}(u,S(u\m)+ S(u\m)\tilde{\gamma}(u,S(u\m),z)) - v_{m}(u,S(u\m) + S(u\m) \eta(u,z))  \\
&\ \phantom{+ \int_{w}^{t}\int_{\I} \frac{1}{M(u)}\Big(} + v'_{m}(u,S(u))S(u\m)(\eta(u,z)-\tilde{\gamma}(u,S(u\m),z))\\
&\ \phantom{+ \int_{w}^{t}\int_{\I} \frac{1}{M(u)}\Big(} + \theta(u,z) \big(v_{m}(u,S(u\m) + S(u\m)\eta(u,z))\\
&\ \phantom{+ \int_{w}^{t}\int_{\I} \frac{1}{M(u)}\Big(} - v_{m}(u,S(u\m)) - S(u\m)\eta(u,z) v'_{m}(u,S(u)) \big) \Big)\ \vartheta(dz) du\ \Bigg| \mathcal{F}_{w} \Bigg] \geq \frac{e_{m}(w)}{M(w)},
\end{split}
\label{eq:E_error}
\end{align}
\noindent
where the last inequality follows from combining Theorem \ref{thm_convex} with Lemma \ref{lem_mon} and because $\theta(t,z) \geq 0$, $d\mathbb{P} \times dt \times \vartheta(dz)$-a.s., whenever $\mathbb{Q} \in \mathcal{Q}_{0}(\p)$. This implies that $e_{m}/M$ is a $\mathcal{Q}_{0}(\p)$-submartingale. Moreover, we easily conclude from \eqref{eq:E_error} that the hedging error has a non-negative expected value at any time $t \in [0,T]$ under each measure $\mathbb{Q} \in \mathcal{Q}_{0}(\p)$ since choosing $w =0$ yields
\begin{equation}
\mathbb{E}_{\mathbb{Q}}\left[\frac{e_{m}(t)}{M(t)}\Big| \mathcal{F}_{0} \right] \geq \frac{e_{m}(0)}{M(0)} = 0.
\label{eq:em_subm}
\end{equation}
If \eqref{eq:em_subm} holds for every $\mathbb{Q} \in \mathcal{Q}_{0}(\p)$, it clearly also holds for the closure of its convex hull. Finally, note that the left-hand side  in \eqref{eq:E_error} is minimal if $\theta=0$. The inequality then is strict (also if the minimum over $\Q$ is taken), unless $v_m$ is linear, or we do not have a jump part and the missspecified model is actually the correct one. In all these cases the hedging error is zero.
\end{proof}

Compared to the reference measure $\mathbb{P}$, two properties of the measures $\mathbb{Q} \in \mathcal{Q}_{0}(\p)$ are noteworthy. First, under the measures $\mathbb{Q} \in \mathcal{Q}_{0}(\p)$, the process $S_{m}^{x}$ is not necessarily Markov. It is only assumed to be Markov under the measure the Delta hedging strategy is calculated. Second, the measures contained in $\mathcal{Q}_{0}(\p)$ have less probability mass on the jump part. From an investor's point of view this goes along with a risk decrease in jump uncertainty.\par
 As the discounted hedging error starts at zero and is a $\mathcal{Q}_{0}(\p)$-submartingale, it is increasing on average. Furthermore, definition (\ref{gooddeal}) and Theorem \ref{theoone} imply that the hedging error is acceptable with $\mathcal{A}:=\{X\in L^2|\inf_{\mathbb{Q} \in \cl(\conv(\mathcal{Q}_{0}(\p)))} \mathbb{E}_{\mathbb{Q}}[X] \geq 0\}$. If the hedging error is not zero, this yields that 
 \begin{align}
 	\label{good-deal}
 	0&>\sup_{\Q\in \cl(\conv(\mathcal{Q}_0(\p)))}\mathbb{E}_{\mathbb{Q}}\left[-\frac{e_{m}(T)}{M(t)} \right]\nonumber\\
 	&\geq \inf_{P_T\in\mathcal{V}}\sup_{\Q\in \cl(\conv(\mathcal{Q}_0(\p)))}\mathbb{E}_{\mathbb{Q}}\left[-\frac{e_{m}(T)}{M(t)}-P_T \right]=  \rho(e_m(T)/M(T)).
 	 \end{align} The last equation follows from the general connection of a coherent risk measure (i.e., a robust expectation) with its acceptance set, see for instance Proposition 4.6 in \cite{fs}. Therefore, 
 \begin{align*}
 0&\notin[ -\rho(e_m(T)/M(T)),\rho(-e_m(T)/M(T))]\\
 &= [-\rho(-h(S(T)/M(T)+v_m(0,x)),\rho(h(S(T))/M(T)-v_m(0,x))].
 \end{align*} 
 Hence, the hedging portfolio $\tilde{P}^{v_{m}'}$ give rise to a good deal of selling $h(S(T))$ for the price $v_m(0,x)$, and in particular the hedging error itself is a good-deal. If on the other hand the hedging error is zero, we have a perfect hedge under the missspecified model. In this case either $h$ is linear or we do not have a jump part and the model is actually correctly specified.\\

  The acceptance set so far is based on the robust expectation with respect to the hull $\cQ_0( \p)$. The next results show that the measure $ \p$ can be chosen very general (while in Section \ref{robustpricing} we will see that actually whole sets of measures instead of a single reference measure $ \p$ can be considered).  
 By Lemma \ref{lemmagirsanov}, we can write 
 \begin{align*}
 \cQ&=\Big\{\Q^{(\psi,\theta)}\sim \p \big|\gamma(t,\tilde{S}_{m}^{x}(t)) \psi(t) + \int_{\I} \tilde{\gamma}(t,\tilde{S}_{m}^{x}(t\m),z) \theta(t,z) \ \vartheta(dz) = 0,\\ &\hspace{3cm} \mbox{ and }\Q^{(\psi,\theta)}\mbox{ satisfies Condition (I)}  \Big\}.
 \end{align*}
 \begin{corollary}
 	Theorem \ref{theoone} holds for any measure $\mathbb{P}\in \tilde{\cQ}$ with  \begin{align*}
 	\tilde{\cQ}&:=\bigg\{\Q^{(\psi,\theta)}\sim\mathbb{P}\bigg|a_t:=\gamma(t,
 	\tilde{S}_{m}^{x}(t)) \psi(t) + \int_{\I} \tilde{\gamma}(t,\tilde{S}_{m}^{x}(t\m),z) \theta(t,z) \ \vartheta(dz)\\
 	&\hspace{1cm} \mbox{ is deterministic and integrable. Furthermore, }\Q^{(\psi,\theta)}\mbox{ satisfies Condition (I)}\bigg\}.
 		\end{align*}
 	 \end{corollary}
\begin{proof}
By the Girsanov theorem, (see Lemma \ref{lemmagirsanov}), a measure $\Q^{(\psi,\theta)}\in \mathcal{Q}_0(\p)$ induces a stock price process following the SDE (\ref{eq:r_stock}) with deterministic drift $r-a,$ and $(W,J)$ replaced by $(W^{\Q^{(\psi,\theta)}},J^{\Q^{(\psi,\theta)}}).$ Hence, the corollary is a consequence of Theorem \ref{theoone} with $r$ replaced by $r-a.$
	\end{proof}
Recall that Assumption \ref{ass_vol} stipulates boundedness conditions on $\gamma$ and $\tilde{\gamma}$. The next proposition shows that, any probability measure whose stochastic logarithm satisfies certain boundedness conditions can be included in the robust expectation constituting the acceptance set.
\begin{proposition}
	\label{bound}
	Suppose that there exists $\delta>0$ and $\varepsilon>0$ such that $\int_{\I}I_{\varepsilon<\tilde{\gamma}(t,\tilde{S}_{m}^{x}(t\m),z)} (|z|^2\wedge 1) \vartheta(dz)>\delta$, or $\int_{\I}I_{\varepsilon<-\tilde{\gamma}(t,\tilde{S}_{m}^{x}(t\m),z)} (|z|^2\wedge 1) \vartheta(dz)>\delta$ $d\p\times d t$ a.s. Then for any $\Q\sim \mathbb{P}$ satisfying Condition(I),
such that the corresponding $\psi$ and $\tilde{\theta}$ (from the Radon-Nikodym derivative of $\Q$) are BMO processes with $\tilde{\theta}<1-\bar{\delta} <1$ satisfying 
	$$\sup_{0\leq t\leq T} a_t:=\sup_{0\leq t\leq T} ||\gamma(t,\tilde{S}_{m}^{x}(t)) \psi(t) + \int_{\I} \tilde{\gamma}(t,\tilde{S}_{m}^{x}(t\m),z) \tilde{\theta}(t,z) \ \vartheta(dz)||_\infty <\infty,$$
	there exists $\mathbb{P}\in \tilde{\Q}$ such that $\Q^{(\psi,\tilde{\theta})}\in \cQ_0(\p)$. 
\end{proposition}
\begin{proof}
In the proof we use ``+'' or ``-'' depending if $\int_{\I}I_{\varepsilon<\tilde{\gamma}(t,\tilde{S}_{m}^{x}(t\m),z)} (|z|^2\wedge 1) \vartheta(dz)>\delta$ or $\int_{\I}I_{\varepsilon<-\tilde{\gamma}(t,\tilde{S}_{m}^{x}(t\m),z)} (|z|^2\wedge 1) \vartheta(dz)>\delta$. Define 
	\begin{equation*}
	\theta(t,z):=\tilde{\theta}(t,z)-
	\frac{a_t\pm[\gamma(t,\tilde{S}_{m}^{x}(t)) \psi(t) + \int_{\I} \tilde{\gamma}(t,\tilde{S}_{m}^{x}(t\m),z) \tilde{\theta}(t,z) \ \vartheta(dz)]}{\int_{\I}I_{\varepsilon<\pm\tilde{\gamma}(t,\tilde{S}_{m}^{x}(t\m),z)} (|z|^2\wedge 1) \vartheta(dz) }\frac{I_{\varepsilon<\pm\tilde{\gamma}(t,\tilde{S}_{m}^{x}(t\m),z)}(|z|^2\wedge 1)}{\pm\tilde{\gamma}(t,\tilde{S}_{m}^{x}(t\m),z)}.
	\end{equation*}
		By the definition of $a_t$ we have then $\theta\leq \tilde{\theta}$. Furthermore, $(\psi,\theta)$ are BMO processes with $\theta$ bounded away from $1$, see also the Appendix. Thus, $\Q^{(\psi,\tilde{\theta})}\in \cQ_0(\p)$ with $\p:=\Q^{(\psi,\theta)}\in \tilde{\cQ}$. 
	\end{proof}
Using our results on the hedging error, we close this section with a corollary stating that systematic overestimations of the volatility and the jump sensitivity lead to a domination of the true contingent claim price by the misspecified price. The reason is that the overestimations make the option more valuable benefiting from the additional volatility. 

\begin{corollary}
Suppose the conditions of Theorem \ref{thm_superhedge_delta} are satisfied. Then the process\\ $(v_{m}(t,S(t))/M(t))_{t \in [0,T]}$ with $v_{m}$ given by \eqref{eq:mis_price} is a supermartingale. In particular,\\ $v_{m}(0,S(0)) \geq \mathbb{E}[h(S(T))/M(T)].$ So the misspecified model also gives a higher (more conservative) claim price.
\end{corollary}
\begin{proof}
Since $e_{m}/M$ is a submartingale under $\mathbb{P}$ and Assumption \ref{ass_martingale} implies that $P^{v_{m}'}/M$ given by \eqref{eq:hedging_portfolio} is a martingale, we conclude from \eqref{eq:hedging_error} that $(v_{m}(t,S(t))/M(t))_{t \in [0,T]}$ is a supermartingale. Consequently, it holds that
$$v_{m}(0,S(0)) = \frac{v_{m}(0,S(0))}{M(0)} \geq \mathbb{E}\left[\frac{v_{m}(T,S(T))}{M(T)}\right] = \mathbb{E}[h(S(T))/M(T)].$$
\end{proof}


\section{Robust pricing}
\label{robustpricing}
\noindent
In this section we analyze pricing such that the hedging error is a good-deal with an acceptance set based on a robust expectation for \emph{general} subsets  $\mathcal{M} \subseteq \mathcal{Q}$. 
By redefining the function $h$ we normalize the interest rate or the drift to zero for convenience for the remainder of the paper. We define the pricing operator at time $t$ of the European claim $h(S(T))$ in the misspecified model by

\begin{equation}
C_{m}^{\mathcal{M}}(t) := \esssup_{\mathbb{Q} \in \mathcal{M}} \mathbb{E}_{\mathbb{Q}}\left[h(S_{m}^{t,x}(T))\ \big| \mathcal{F}_{t} \right].
\label{eq:gen_C_def}
\end{equation}
\noindent
In the sequel we omit the superscript $\mathcal{M}$ unless there is ambiguity. A non-exhaustive list of examples for the choice of $\mathcal{M}$ leading to a well known robust pricing operator $C_{m}$ is the following:

\begin{enumerate}
	\item \textsc{Complete market: } In a complete market 
	$\cQ$ has just one element, so that $\mathcal{M}:=\{\mathbb{Q}\}$ where $\Q$ is the unique equivalent local martingale measure. 
	\item  \textsc{Sharpe ratio:} Pricing contingent claims by only requiring the absence of arbitrage leads to a relatively wide range of prices; \cite{gdb} outline that prices beyond a certain benchmark correspond to unreasonably good deals. To overcome this, they suggest to narrow the no-arbitrage bounds by imposing bounds on the \emph{Sharpe ratio} in a Brownian setting. In \cite{slinko} these results are extended to L\'evy processes. In this case the set of acceptable positions $\mathcal{A}$ consists of all position with a Sharpe ratio larger than a constant, say $k$. That is
	$$ \mathcal{A}=\Big\{X\in L^2\Big| \frac{\mathbb{E}[X]-\pi(X)}{\sigma(X)}>k\Big\}=\{X\in L^2 |\inf_{\mathbb{Q} \in \mathcal{M}}\mathbb{E}[X]\geq 0\}   ,$$
	where $\pi(X)$ is the observed market price of $X$, and $\mathcal{M}$ is given by
	\begin{align*}
	\mathcal{M} := \left\{\mathbb{Q} \in \mathcal{Q}\Big| \ \psi(t)^{2} + \int_{\I} \theta(t,z)^{2}\ \vartheta(dz) \leq B,\ \text{for Lebesgue-almost all}\ t \in [0,T]\right\}, 
	\end{align*}
	for some constant $B>0$ depending on $k$. The last equation is shown in \cite{gdb} using the Hansen-Jagannathan bound. $C_m(t)$ in \eqref{eq:gen_C_def} is then also called upper good-deal bound and any price at time $t$ above $C_m(t)$ gives rise to a good-deal (i.e., allowing for portfolios with price zero which are in $\mathcal{A}$).\\
\item\textsc{Strictly Acceptable Opportunities:} In \cite{carr2001pricing} the acceptance set $\mathcal{A}$ is given by the set of all financial positions $X$ such that
$$\mathcal{A}=\{X\in L^2 | \mathbb{E}_{\mathbb{Q}_i}[X]\geq 0, \,\,i=1,\ldots,m\} ,  $$ 
for some test measures $\Q_1,\ldots,\Q_m,$ which for instance could be pricing measures based on different marginal utilities of different agents. Note that in an equilibrium model every rational agent is willing to accept at least some scenarios where possible losses can occur. $\mathcal{M}$ in this case is the set of these test measures, and a price higher than \eqref{eq:gen_C_def} leads to a (strictly) acceptable oppertunity. \cite{carr2001pricing} prove fundamental theorems of asset pricing replacing the traditional notion of arbitrage by their notion of acceptability. These results were generalized for instance by \cite{cherny2008pricing}.\\
	\item \textsc{Ball scenarios:} A measure $\mathbb{P}$ qualifies to be a reference risk-neutral measure if it is plausible from an investor's point of view. To obtain an equally reasonable price range of the claim $h(S(T))$ in \eqref{eq:gen_C_def}, one can consider market models $\mathbb{Q} \in \mathcal{Q}$ in a small ball around the reference model $\mathbb{P}.$ Noting that the choice $\psi = \theta = 0$ in \eqref{eq:girsanov} leads to $\mathbb{P}$, one specifies
	\begin{align*}
	\mathcal{M} := \{&\mathbb{Q} \in \mathcal{Q} |\ |\psi(t)| \leq B_{1}, \ |\theta(t,z)| \leq B_{2}, \  d\mathbb{P}\times dt\times \vartheta(dz)\text{-a.s.}\},
	\end{align*}
	for some constants $B_{1},B_{2} > 0$. The set of acceptable position is then defined as
	$$\mathcal{A}=\{X \in L^2|\inf_{\mathbb{Q} \in \mathcal{M}}\mathbb{E}_{\mathbb{Q}}[X]\geq 0\}.$$

	
\end{enumerate} 
Denote by $\mathcal{P}$ the predictable $\sigma$-algebra on $[0,T] \times \Omega$ w.r.t. $(\mathcal{F}_{t}).$
\begin{remark} \label{rem}
If the set $\mathcal{M}$ is of the form
$$\mathcal{M} = \{\mathbb{Q} \in \mathcal{Q}| (\psi(t),\theta(t,\cdot)) \in H(t,\omega)\},$$
for a compact $\mathcal{P} \times \mathcal{B}(\mathbb{R}) \times \mathcal{B}(L^{2}(\vartheta(dz))$-measurable set $H$, there exists $\bar{\mathbb{Q}} \in \mathcal{M}$ such that for all $t \in [0,T]: C_{m}(t,x) = \mathbb{E}_{\bar{\mathbb{Q}}}[h(S_{m}^{t,x}(T))\ | \mathcal{F}_{t}].$ 
\end{remark}

An assumption ensuring the Markovian structure of the robust price under the misspecified model is needed. To this end, we define the set of probability measures $\mathcal{A}$ leading to a Markov process $S_{m}^{x}$ as 
\begin{equation}
\mathcal{U} := \{\mathbb{Q}|\ \mathbb{Q}\ \text{is a probability measure on}\ (\Omega,\mathcal{F})\ \text{and}\  S_{m}^{x}\ \text{is Markov w.r.t.}\ \mathbb{Q}\}.
\label{eq:A}
\end{equation}

\begin{assumption}
We assume that we can restrict ourselves to Markov processes in \eqref{eq:gen_C_def}, i.e.,  $C_{m}(t) = \esssup_{\mathbb{Q} \in \mathcal{M} \cap \mathcal{U}} \mathbb{E}_{\mathbb{Q}}\left[h(S_{m}^{t,x}(T))\ \big| \mathcal{F}_{t} \right].$ 
\label{ass_markov_gen_C}
\end{assumption}
\noindent
Assumption \ref{ass_markov_gen_C} is typically satisfied in the examples above if the stock process is Markov under some benchmark model because in this case the set $H$ of Remark \ref{rem} has a deterministic boundary such that the supremum is attained in some $\mathbb{Q}$ maintaining an initial Markovian structure of $S_m$. Due to Assumption \ref{ass_markov_gen_C}, we can express the essential supremum in \eqref{eq:gen_C_def} as deterministic function of time and $S_{m}^{x}(t)$, i.e., 
\begin{equation}
C_{m}(t) = C_{m}(t,S_{m}^{x}(t)),
\end{equation}
for $C_{m}:[0,T] \times \mathbb{R}_{+} \to \mathbb{R}.$ We deduce another corollary from Theorem \ref{thm_convex}.
\begin{corollary}
Suppose the conditions of Theorem \ref{thm_convex} and Assumption \ref{ass_markov_gen_C} are satisfied. Then the function $C_{m}$ is convex and has bounded one-sided derivatives in the second variable. 
\label{cor_convex}
\end{corollary}
\noindent
Observe that the price $C_{m}$ at time $T$ is equal to the value of the claim itself because 
\begin{equation}
C_{m}(T,S_{m}^{T,S_{m}(T)}(T)) = \esssup_{\Q \in \mathcal{M} \cap \mathcal{U}} \mathbb{E}_{\mathbb{Q}}\left[h(S_{m}^{T, S_{m}(T)}(T)) | \mathcal{F}_{T}\right] = h(S_{m}(T)).
\label{eq:cost_final_gen}
\end{equation}
\noindent
The following regularity assumption is needed for later applications of \emph{It\^{o}'s lemma}:

\begin{assumption}
We assume that either $C_{m} \in C^{1,2}$ \emph{or}, if $\gamma \equiv 0$ and $\int_{\I} |z|\ \vartheta(dz) < \infty,$ that $C_{m} \in C^{1,1}$ \emph{or}, if $\gamma \equiv 0$ and  $\vartheta(\I) < \infty,$ that $C_{m}$ is locally Lipschitz in $t$. 
\label{ass_ito_C}
\end{assumption}

\begin{remark}
If $\gamma \equiv 0$ and $\int_{\I} |z|\ \vartheta(dz) < \infty$, then for It\^{o}'s formula (and therefore all results in the sequel) to hold, it is sufficient that $C_{m} \in C^{1,1}$, see for instance Theorem 4.2 in \cite{kyprianou2014fluctuations}.
\end{remark}
\noindent
The next lemma shows that the conditions enforced in the third part of Assumption \ref{ass_ito_C} indeed yield \emph{It\^{o}'s formula}. Recall that $C'_{m,+}$ denotes the right-hand derivative of $C_{m}.$
\begin{proposition}
Suppose that $\vartheta(\I) < \infty$ and $\gamma \equiv 0.$ If $C_{m}$ is continuously differentiable in time, then It\^{o}'s formula holds with $C_{m}' = C'_{m,+}.$
\label{lem_ito}
\end{proposition}
\begin{proof}
Fix $t \in [0,T]$ and define
\begin{align*}
\lambda(t) &:= - \int_{\I} \tilde{\gamma}(t,S_{m}^{x}(t),z)\ \vartheta(dz), \\
S^{J}(t) &:= \int_{0}^{t} \int_{\I} \ln(1+\tilde{\gamma}(u,S_{m}^{x}(u),z)) \ J(du,dz).
\end{align*}
\noindent
Assumption \ref{ass_vol} and the condition $\vartheta(\I) < \infty$ ensure that $\lambda(t)$ is bounded for every $t \in [0,T].$ Suppressing as usual the dependence on $\omega$, we can define some function $f$ by $f(t,s^{J}) = C_{m}\left(t,\exp \left(\int_{0}^{t} \lambda(u)\ du + s^{J} \right)\right).$ Observe that $f$ is locally Lipschitz continuous in $t$. Thus, it is almost everywhere (a.e.) differentiable in $t$ by the \emph{Lebesgue theorem} (as it is absolutely continuous) and equal to the integral of its derivative. Corollary \ref{cor_convex} implies that $C_{m}$ is a.e. differentiable in its second component and since its derivative is equal to $C'_{m,+}$, the derivative of $f$ w.r.t. $t$ is a.e. equal to 
\begin{equation}
\dot{f}(t,S^{J}(t)) = \dot{C}_{m}(t,S_{m}^{x}(t)) + S_{m}^{x}(t) \lambda(t) C'_{m,+}(t,S_{m}^{x}(t)).
\label{eq:f}
\end{equation}
As $\vartheta(\I) < \infty$ and $\gamma \equiv 0$, $J(dt,dz)$ corresponds to a compound Poisson process whose jump times we denote by $(\tau_{n})_{n \geq 1}$ with $\tau_{0} = 0.$ It is
\begingroup
\allowdisplaybreaks
\begin{align*}
&C_{m}(t,S_{m}^{x}(t))\\ &= f(t,S^{J}(t)) \\
& = \sum_{n \geq 1, \tau_{n} \leq t} \left(\int_{\tau_{n-1}}^{\tau_{n}}\dot{f}(u,S^{J}(u\m)) \ du + f(\tau_{n},S^{J}(\tau_{n})) - f(\tau_{n},S^{J}(\tau_{n-1})) \right) \\
&= \int_{0}^{t} \dot{f}(u,S^{J}(u\m)) \ du + \sum_{0 \leq u \leq t} \left(f(u,S^{J}(u)) - f(u,S^{J}(u\m)) \right) \\
&\stackrel{\eqref{eq:f}}{=} \int_{0}^{t} \dot{C}_{m}(t,S_{m}^{x}(t)) \ du + \int_{0}^{t} C'_{m,+}(u,S_{m}^{x}(u\m)) \ dS_{m}^{x}(u)\\
&\ \ + \int_{0}^{t} \int_{\I} \Big(C_{m}(u,S_{m}^{x}(u)) - C_{m}(u,S_{m}^{x}(u\m)) - C'_{m,+}(u,S_{m}^{x}(u\m))\ \Delta S_{m}^{x}(u,z)\Big)\ J(du,dz)  ,
\end{align*}
\endgroup
where the first equation holds because $S^{J}$ is constant between the jumps and by the \emph{Lebesgue theorem} explained before. This gives \emph{It\^{o}'s formula}.
\end{proof}
\begin{remark}
In the sequel we will for the ease of exposition with a slight abuse of notation denote $C_{m}' := C_{m,+}'$ and $C_{m}''\cdot \gamma^{2} = C_{m}''\cdot \sigma^{2}:= 0$ in case $\gamma = \sigma = 0$ even if $C_{m}''$ formally does not exist. 
\end{remark}

Suppose some investor chooses to follow the trading strategy $C_{m}' = (C_{m}'(t,\cdot))_{t \in [0,T]}$ that is computed with respect to the misspecified model and trades in the physical stock $S$. Then the corresponding self-financing hedging portfolio is given by

\begin{equation}
P^{C_{m}'}(t,S(t)) = C_{m}(0,x) + \int_{0}^{t} C_{m}'(u,S(u\m)) \ dS(u).
\label{eq:hedging_P}
\end{equation}
\noindent
We denote by $E^{C_{m}'} = (E^{C_{m}'}(t))_{t \in [0,T]}$ the hedging error induced, and formally define it as

\begin{equation}
E^{C_{m}'}(t) := P^{C_{m}'}(t,S(t)) - C_{m}(t,S(t)).
\label{eq:E_C}
\end{equation}
\noindent
Hence, $E^{C_{m}'}(T)$ gives the difference how far off the terminal value of the hedging portfolio $P^{C_{m}'}$ is from the payoff $h(S(T)).$ The following integrability condition is needed for the proof of Theorem \ref{thm_main} below; we remark that it is for instance fulfilled if $\tilde{\gamma}$ is (uniformly) square-integrable in a neighborhood of zero. 

\begin{assumption}
In addition to Assumption \ref{ass_vol}, suppose that for all $(t,s) \in [0,T] \times \mathbb{R}_{+}$ there exists some $B > 0$ such that
\begin{align*}
\int_{\I} \ln(1+\tilde{\gamma}(t,s,z))^{2}\ \vartheta(dz) \leq B,\ \int_{\I} \tilde{\gamma}(t,s,z) - \ln(1+\tilde{\gamma}(t,s,z)) \ \vartheta(dz) \leq B.
\end{align*}
\label{ass_vol_2}
\end{assumption}
Let us denote $\cQ_0(\mathcal{M}):=\bigcup_{\Q\in \mathcal{M}}\mathcal{Q}_0(\Q)$, see (\ref{hull}).
The following theorem analyzes the hedging error $E^{C_{m}'}$ and states that it is a submartingale w.r.t. each measure $\mathbb{Q} \in \cl(\conv(\cQ_0(\mathcal{M}))$ if the volatility and the jump sensitivity are systematically overestimated. 
\begin{theorem} \label{thm_rob_E}
Suppose Assumption \ref{ass_martingale}, Assumption \ref{ass_markov_gen_C}, Assumption \ref{ass_ito_C} and Assumption \ref{ass_vol_2} are satisfied. Consider the hedging portfolio $P^{C_{m}'}$ given by \eqref{eq:hedging_P} and the corresponding hedging error $E^{C_{m}'}$ specified by \eqref{eq:E_C}. If 
\begin{equation*}
\sigma(t) \leq \gamma(t,S(t)) \ \text{and} \ \emph{sgn}(\tilde{\gamma}(t,S(t),z)-\eta(t,z)) = \emph{sgn}(\eta(t,z)),
\end{equation*}
$d\mathbb{P} \times dt$ and $d\mathbb{P} \times dt \times \vartheta(dz)$-a.s., then the hedging error $E^{C_{m}'}$ is an $\cl(\conv(\cQ_0(\mathcal{M})))$-submartingale. In particular, $\inf_{\mathbb{Q} \in \cl(\conv(\cQ_0(\mathcal{M})))} \mathbb{E}_{\mathbb{Q}}[E^{C_{m}'}(t)] \geq 0$ for all $t \in [0,T].$
\label{thm_main}
\end{theorem}
\begin{proof} The proof is divided into two steps. Recall that the monotonicity of the mapping $x \mapsto S_{m}^{x}(t)$ for any $t$ established in \textsc{Step 1} of the proof of Theorem \ref{thm_convex} holds up to some set $\mathcal{N}_{1}$ of measure zero. \\

\noindent
\textsc{Step 1} is to show that for fixed $(t,\omega) \in [0,T] \times \Omega \setminus \mathcal{N}_{1}$ it holds that $\lim_{x \downarrow 0} S_{m}^{x}(t) = 0$ and $\lim_{x \uparrow \infty} S_{m}^{x}(t) = \infty.$\\
Regarding the first statement, let $(x_{k})_{k \in \mathbb{N}}$ be a strictly positive sequence satisfying $\lim_{k \to \infty} x_{k} = 0.$ Observe that
$$\lim_{k \to \infty} \mathbb{E}[|S_{m}^{x_{k}}(t)|] = \lim_{k \to \infty} \mathbb{E}[S_{m}^{x_{k}}(t)] = \lim_{k \to \infty} S_{m}^{x_{k}}(0) = \lim_{k \to \infty} x_{k} =0,$$
that is, $S_{m}^{x_{k}}(t)$ converges to $0$ in $L^{1}$. Hence, there exists a subsequence $(x'_{k})_{k \in \mathbb{N}}$ along which $S_{m}^{x'_{k}}(t)$ converges to $0$ a.s. As the mapping $x \mapsto S_{m}^{x}(t)$ is monotonously increasing, the aforementioned subsequence can be assumed to coincide with $(x_{k})_{k \in \mathbb{N}}.$\\

\noindent
Regarding the second statement, let $(x_{k})_{k \in \mathbb{N}}$ be an arbitrary sequence such that $\lim_{k \to \infty} x_{k} = \infty.$ By contradiction assume that 
\begin{equation}
\mathbb{P}\left[\sup_{x > 0} S_{m}^{x}(t) < \infty \right] = \mathbb{P}\left[\lim_{k \to \infty} S_{m}^{x_{k}}(t) < \infty \right] > 0.
\label{eq:contra_P}
\end{equation}
 Note that the previous equality holds by monotonicity. Define
\begin{align*}
B&:=\Big\{\omega \in \Omega \setminus \mathcal{N}_{1} \Big|\ \lim_{k \to \infty}\Big( \int_{0}^{t} \gamma(u,S_{m}^{x_{k}}(u))\ dW(u)\\
&\ \ + \int_{0}^{t} \int_{\I} \ln(1+\tilde{\gamma}(u,S_{m}^{x_{k}}(u\m),z)) \ \tilde{J}(du,dz) \Big)= - \infty\Big\}.
\end{align*}
\noindent
Then we see that
\begingroup
\allowdisplaybreaks
\begin{align*}
&\mathbb{P}\left[\sup_{x > 0} S_{m}^{x}(t) < \infty \right] \\
&=\mathbb{P}\left[\lim_{k \to \infty} S_{m}^{x_{k}}(t) < \infty \right] \\
&=\mathbb{P}\Big[\lim_{k \to \infty}x_{k} \exp\Big(\int_{0}^{t} \gamma(u,S_{m}^{x_{k}}(u))\ dW(u) - \frac{1}{2} \int_{0}^{t} \gamma(u,S_{m}^{x_{k}}(u))^{2}\ du \\
&\ \ + \int_{0}^{t} \int_{\I} \ln(1+\tilde{\gamma}(u,S_{m}^{x_{k}}(u\m),z)) \ \tilde{J}(du,dz)\\
&\ \ - \int_{0}^{t}\int_{\I} \left(\tilde{\gamma}(u,S_{m}^{x_{k}}(u\m),z)-\ln(1+\tilde{\gamma}(u,S_{m}^{x_{k}}(u\m),z)) \right) \ \vartheta(dz)du \Big) < \infty \Big] \\
&\leq \mathbb{P}\Big[\lim_{k \to \infty} \Big(\int_{0}^{t} \gamma(u,S_{m}^{x_{k}}(u))\ dW(u) - \frac{1}{2} \int_{0}^{t} \gamma(u,S_{m}^{x_{k}}(u))^{2}\ du \\
&\ \ + \int_{0}^{t} \int_{\I} \ln(1+\tilde{\gamma}(u,S_{m}^{x_{k}}(u\m),z)) \ \tilde{J}(du,dz)\\
&\ \ - \int_{0}^{t}\int_{\I} \left(\tilde{\gamma}(u,S_{m}^{x_{k}}(u\m),z) - \ln(1+\tilde{\gamma}(u,S_{m}^{x_{k}}(u\m),z))\right) \ \vartheta(dz)du \Big) = - \infty \Big] = \mathbb{P}(B).
\end{align*}
\endgroup
The last equality is justified by Assumption \ref{ass_vol_2}. Observe that

\begin{align*}
0 \leq X_{k} &:= \Big|\int_{0}^{t} \gamma(u,S_{m}^{x_{k}}(u))\ dW(u) + \int_{0}^{t} \int_{\I} \ln(1+\tilde{\gamma}(u,S_{m}^{x_{k}}(u\m),z)) \ \tilde{J}(du,dz) \Big|\cdot \mathbbm{1}_{B} \\
&\leq \Big|\int_{0}^{t} \gamma(u,S_{m}^{x_{k}}(u))\ dW(u) + \int_{0}^{t} \int_{\I} \ln(1+\tilde{\gamma}(u,S_{m}^{x_{k}}(u\m),z)) \ \tilde{J}(du,dz) \Big| \in L^{1}.
\end{align*}
\noindent
Note that again Assumption \ref{ass_vol_2} ensures that the previous term is in $L^{1}.$ Then, by the definition of $B$, it holds that 
$$ \lim_{k \to \infty}
X_{k}=
\begin{cases}
\infty,\ &\text{on}\ B,\\
0, &\text{else.}\end{cases}
$$
It follows from \eqref{eq:contra_P} that $\mathbb{P}(B) > 0,$ so we clearly see that
\begin{equation}
\lim_{k \to \infty} \mathbb{E}[X_{k}] = \infty.
\label{eq:contra}
\end{equation}
However, as
\begingroup
\allowdisplaybreaks
\begin{align*}
&\mathbb{E}[X_{k}^{2}] \\
&=\mathbb{E}\left[\Big|\int_{0}^{t} \gamma(u,S_{m}^{x_{k}}(u))\ dW(u) + \int_{0}^{t} \int_{\I} \ln(1+\tilde{\gamma}(u,S_{m}^{x_{k}}(u\m),z)) \ \tilde{J}(du,dz) \Big|^{2}\cdot \mathbbm{1}_{B} \right] \\
&\leq \mathbb{E}\left[\left(\int_{0}^{t} \gamma(u,S_{m}^{x_{k}}(u))\ dW(u) + \int_{0}^{t} \int_{\I} \ln(1+\tilde{\gamma}(u,S_{m}^{x_{k}}(u\m),z)) \ \tilde{J}(du,dz) \right)^{2}\right] \\
&\leq 2 \mathbb{E}\left[\left(\int_{0}^{t} \gamma(u,S_{m}^{x_{k}}(u))\ dW(u) \right)^{2} \right] +2 \mathbb{E}\left[\left(\int_{0}^{t} \int_{\I} \ln(1+\tilde{\gamma}(u,S_{m}^{x_{k}}(u\m),z)) \ \tilde{J}(du,dz) \right)^{2} \right] \\
&= 2 \mathbb{E}\left[\int_{0}^{t} \gamma(u,S_{m}^{x_{k}}(u))^{2}\ du \right] +2 \mathbb{E}\left[\int_{0}^{t} \int_{\I} \ln(1+\tilde{\gamma}(u,S_{m}^{x_{k}}(u\m),z))^{2} \ \vartheta(dz)du \right] \\
&\leq 2T(B_{1}^{2} + B_{2}^{2}),
\end{align*}
\endgroup
for constants $B_{1},B_{2} > 0,$ the \emph{Cauchy-Schwartz inequality} implies that
$$\mathbb{E}[X_{k}] \leq \sqrt{\mathbb{E}[X_{k}^{2}]} \leq \sqrt{2T(B_{1}^{2} + B_{2}^{2})},$$
which contradicts \eqref{eq:contra}. Hence, we conclude that
$$\mathbb{P}\left[\sup_{x>0} S_{m}^{x}(t) < \infty \right] = 0 \Rightarrow \mathbb{P}\left[\sup_{x>0} S_{m}^{x}(t) = \infty \right] = 1.$$

\noindent
\textsc{Step 2} is to establish the claim of the theorem. An application of \emph{It\^{o}'s lemma}, compensating the jump-integral and \eqref{eq:girsanov} yield

\begin{align}
&C_{m}(t,S_{m}^{x}(t))\notag \\
\begin{split}
&= C_{m}(0,x) + \int_{0}^{t} \dot{C}_{m}(u,S_{m}^{x}(u)) \ du + \int_{0}^{t} C_{m}'(u,S_{m}^{x}(u\m)) \ dS_{m}^{x}(u) \\
&\ \ + \frac{1}{2} \int_{0}^{t} C_{m}''(u,S_{m}^{x}(u)) \ S_{m}^{x}(u)^{2}\ \gamma(u,S_{m}^{x}(u))^{2}\ du\\
&\ \ + \int_{0}^{t} \int_{\I} \Big(C_{m}(u,S_{m}^{x}(u\m) + \Delta S_{m}^{x}(u,z)) - C_{m}(u,S_{m}^{x}(u\m))\\
&\ \ \phantom{+ \int_{0}^{t} \int_{\I} \Big(} - \Delta S_{m}^{x}(u,z) C_{m}'(u,S_{m}^{x}(u\m)) \Big)\ \tilde{J}_{\mathbb{Q}}(du,dz) \\
&\ \ + \int_{0}^{t} \int_{\I} \Big(C_{m}(u,S_{m}^{x}(u\m) + \Delta S_{m}^{x}(u,z))
- C_{m}(u,S_{m}^{x}(u\m))\\
&\ \ \phantom{+ \int_{0}^{t} \int_{\I} \Big(}- \Delta S_{m}^{x}(u,z) C_{m}'(u,S_{m}^{x}(u\m)) \Big) (1-\theta(u,z))\ \vartheta(dz)du.
\end{split}
\label{eq:gen_C}
\end{align}

%
%
\noindent
Note that the compensated jump integral in \eqref{eq:gen_C} gives rise to a true martingale under each measure $\mathbb{Q} \in \mathcal{M}$ due to Corollary \ref{cor_convex} and \eqref{eq:replace}. Now it follows directly from \eqref{eq:gen_C_def} that $C_{m}$ is a supermartingale. Hence, the predictable process $A=(A(t))_{t \in [0,T]}$ given by

\begin{align*}
A(t) &= -\int_{0}^{t} \dot{C}_{m}(u,S_{m}^{x}(u)) \ du - \frac{1}{2} \int_{0}^{t} C_{m}''(u,S_{m}^{x}(u)) \ S_{m}^{x}(u)^{2}\ \gamma(u,S_{m}^{x}(u))^{2}\ du \\
& \ \ - \int_{0}^{t} \int_{\I} \Big(C_{m}(u,S_{m}^{x}(u\m) + \Delta S_{m}^{x}(u,z))
- C_{m}(u,S_{m}^{x}(u\m))\\
&\ \ \phantom{- \int_{0}^{t} \int_{\I} \Big(} - \Delta S_{m}^{x}(u,z) C_{m}'(u,S_{m}^{x}(u\m) \Big) (1-\theta(u,z))\ \vartheta(dz)du
\end{align*}
 must be increasing. Let $0 \leq w < t \leq T.$ We conclude from the \emph{Lebesgue differentiation theorem} (cf. \cite{rudin}, Chapter 7)  that the following inequality holds Lebesgue-a.s.:
\begingroup
\allowdisplaybreaks
\begin{align}
0 &\leq \lim_{w \to t} \frac{A(t)-A(w)}{t-w}\notag \\
&= \lim_{w \to t} \frac{1}{t-w} \Bigg(-\int_{w}^{t} \dot{C}_{m}(u,S_{m}^{x}(u)) \ du - \frac{1}{2} \int_{w}^{t} C_{m}''(u,S_{m}^{x}(u)) \ S_{m}^{x}(u)^{2}\ \gamma(u,S_{m}^{x}(u))^{2}\ du\notag \\
&\ \ - \int_{w}^{t} \int_{\I} \Big(C_{m}(u,S_{m}^{x}(u) +S_{m}^{x}(u)\tilde{\gamma}(u,S_{m}^{x}(u),z))
- C_{m}(u,S_{m}^{x}(u))\notag \\
&\ \ \phantom{- \int_{w}^{t} \int_{\I} \Big(} - S_{m}^{x}(u)\tilde{\gamma}(u,S_{m}^{x}(u),z) C_{m}'(u,S_{m}^{x}(u)) \Big) (1-\theta(u,z))\ \vartheta(dz)du \Bigg)\notag \\
\begin{split}
&= -\Big( \dot{C}_{m}(t,S_{m}^{x}(t)) + \frac{1}{2} C_{m}''(t,S_{m}^{x}(t)) \ S_{m}^{x}(t)^{2}\ \gamma(t,S_{m}^{x}(t))^{2}\\
&\ \  + \int_{\I} \Big(C_{m}(t,S_{m}^{x}(t) + S_{m}^{x}(t)\tilde{\gamma}(t,S_{m}^{x}(t),z)) - C_{m}(t,S_{m}^{x}(t))\\
&\ \ \phantom{+ \int_{\I} \Big(} - S_{m}^{x}(t)\tilde{\gamma}(t,S_{m}^{x}(t),z)\ C_{m}'(t,S_{m}^{x}(t)) \Big) (1-\theta(t,z))\ \vartheta(dz) \Big).
\end{split}
\label{eq:A_gen} 
\end{align}
\endgroup
\noindent

\noindent
Let us remark that the right-hand side of (\ref{eq:A_gen}) is increasing in $\theta$. In particular, (\ref{eq:A_gen}) holds for all $\tilde{\theta}$ corresponding to a measure $\Q^{(\psi,\tilde{\theta})}\in \cQ_0(\mathcal{M}).$
Next, observe that \eqref{eq:A_gen} must hold for every $S_{m}^{x}(t).$ Since we proved in \textsc{Step 1} that the image of the mapping $x \mapsto S_{m}^{x}(t)$ is the entire positive reals, we see that \eqref{eq:A_gen} holds true when replacing $S_{m}^{x}(t)$ by some arbitrary $s \in \mathbb{R}_{+}.$ This implies that we can substitute $S_{m}^{x}(t)$ in particular by the true stock price $S(t).$ Using the definition of $E^{C_{m}'}$ from \eqref{eq:E_C}, we obtain from \eqref{eq:gen_C} for any $\Q^{(\psi,\tilde{\theta})}\in \cQ_0(\mathcal{M})$ that
\begingroup
\allowdisplaybreaks
\begin{align*}
E^{C_{m}'}(t) &= P^{C_{m}'}(t,S(t)) - C_{m}(t,S(t))\\
&= -\int_{0}^{t} \dot{C}_{m}(u,S(u)) \ du - \frac{1}{2} \int_{0}^{t} C_{m}''(u,S(u)) \ S(u)^{2}\ \sigma(u)^{2}\ du \\
&\ \ - \int_{0}^{t} \int_{\I} \Big(C_{m}(u,S(u\m) + \Delta S(u,z))
- C_{m}(u,S(u\m))\\
&\ \ \phantom{- \int_{0}^{t} \int_{\I} \Big(} - \Delta S(u,z) C_{m}'(u,S(u\m)) \Big) (1-\tilde{\theta}(u,z))\ \vartheta(dz)du \\
&\ \ - \int_{0}^{t} \int_{\I} \Big(C_{m}(u,S(u\m) + \Delta S(u,z)) - C_{m}(u,S(u\m)) \\
&\ \ \phantom{- \int_{0}^{t} \int_{\I} \Big(} - \Delta S(u,z) C_{m}'(u,S(u\m)) \Big)\ \tilde{J}_{\mathbb{Q}}(du,dz). 
\end{align*}
\endgroup
\noindent
Fix again $0 \leq w < t \leq T.$ Since $C_{m}'$ is uniformly bounded according to Corollary \ref{cor_convex}, we can deduce from \eqref{eq:replace} that 
\begingroup
\allowdisplaybreaks
\begin{align}
\label{submartingale}
&\mathbb{E}_{\mathbb{Q}}[E^{C_{m}'}(t)|\mathcal{F}_{w}]\nonumber\\
& \geq E^{C_{m}'}(w) - \mathbb{E}_{\mathbb{Q}}\Bigg[\int_{w}^{t} \dot{C}_{m}(u,S(u)) \ du + \frac{1}{2} \int_{w}^{t} C_{m}''(u,S(u)) \ S(u)^{2}\ \gamma(u,S(u))^{2} \ du \nonumber \\
&\ \ + \int_{w}^{t} \int_{\I} \Big(C_{m}(u,S(u\m) + S(u\m) \tilde{\gamma}(u,S(u\m),z))
- C_{m}(u,S(u\m))\nonumber\\
&\ \ \phantom{+ \int_{w}^{t} \int_{\I} \Big(} - S(u\m) \tilde{\gamma}(u,S(u\m),z) C_{m}'(u,S(u\m)) \Big) (1-\tilde{\theta}(u,z))\ \vartheta(dz)du\ \Bigg|\ \mathcal{F}_{w} \Bigg] \geq E^{C_{m}'}(w),
\end{align}
\endgroup
whereby the first inequality follows from Corollary \ref{cor_convex} (giving the convexity of $C_m$) combined with Lemma \ref{lem_mon} and the second one from \eqref{eq:A_gen}. Hence, $E^{C_{m}'}$ is a submartingale under each measure $\mathbb{Q} \in \cQ_0(\mathcal{M}).$ 
\end{proof}

\begin{remark}
In contrast to \cite{kunita2004stochastic}, who uses Kolmogorov's criterion, we show the bijectivity of stochastic flows under the milder assumptions that the derivatives of the coefficients of $S_{m}^{x}$ are locally Lipschitz  (instead of globally Lipschitz) and that the mapping $x \mapsto x + \tilde{\rho}(t,x,z)$ is homeomorphic only on $\mathbb{R}_{+}$ (instead of being homeomorphic on $\mathbb{R}$) for fixed $(t,z) \in [0,T] \times \I.$ Further results on stochastic flows for non-Lipschitz coefficients for multi-dimensional SDEs in a Brownian and a Brownian-Poisson filtration can be found in \cite{zhang2005homeomorphic} and \cite{qiao2008homeomorphism}.  
\end{remark}

Definition (\ref{gooddeal}) and  Theorem \ref{thm_rob_E} imply that the hedging error is acceptable with $$\mathcal{A}:=\bigg\{X\in L^2| \inf_{\mathbb{Q} \in \cl(\conv(\mathcal{Q}_{0}(\mathcal{M})))} \mathbb{E}_{\mathbb{Q}}\left[X \right] \geq 0\bigg\}.$$ 
Furthermore, we note that the right-hand side in (\ref{eq:A_gen} ) is minimal for $\theta$'s with a corresponding $\Q$ being in $\mathcal{M}$. Hence, if we additionally assume that $\mathcal{M}$ is weakly compact, then the infimum in $\inf_{\mathbb{Q} \in \cl(\conv(\mathcal{Q}_{0}(\mathcal{M})))} \mathbb{E}_{\mathbb{Q}}\left[e_m(T) \right] $ is attained. Now if $h$ is linear, or the model is correctly specified and we do not have a jump part, the hedging error is zero and we have a perfect hedge. Otherwise, the inequality in (\ref{submartingale}) is strict for the $\Q\in \mathcal{M}$ attaining the infimum and therefore
\begin{align*}0 & \notin [ -\rho(e_m(T)/M(T)),\rho(-e_m(T)/M(T))]\\
&= [-\rho(-h(S(T)/M(T)+v_m(0,x)),-\rho(h(S(T)/M(T)-v_m(0,x))]     ,\end{align*}
see (\ref{gooddeal}) (or (\ref{good-deal})). In particular, the hedging error is either zero or a good-deal.


\section{Robust superhedging}

In this section we turn to the computation and properties of the superhedging strategy. We first characterize the superhedging price function by comparing it to the theory outlined in the previous section. Subsequently we discuss several special cases in Theorem \ref{thm_superhedge}. We define \emph{the cost of super-replication} at time $t \in [0,T)$, say $\bar{C}_{m}(t),$ by 

\begin{align}
\begin{split}
\bar{C}_{m}(t) := \essinf\Big\{c(t) \in L^{2}(\mathcal{F}_{t}) :&\text{\ there exists a self-financing trading strategy} \ y \ \text{such that} \\
&\mathbb{P}\left(c(t) + \int_{t}^{T} y(u\m)\ dS_{m}^{t,x}(u) \geq h(S_{m}^{t,x}(T)) \right) =1\Big\},
\end{split}
\label{eq:price1}
\end{align}
\noindent
i.e., it is the \emph{least} amount of money needed to setup a superhedge. As $h'$ is uniformly bounded, the superhedging price process $\bar{C}_{m}$ is well-defined. In \cite{kramkov} it is shown that

\begin{equation}
\bar{C}_{m}(t) = \esssup_{\Q \in \mathcal{Q}_{em}} \mathbb{E}_{\mathbb{Q}}\left[h(S_{m}^{t,x}(T))\Big|\mathcal{F}_{t}\right].
\label{eq:kramkov1}
\end{equation}
\noindent
The equality of \eqref{eq:price1} and \eqref{eq:kramkov1} is known as \emph{superhedging duality}. Recall the definition of the set $\mathcal{U}$ in \eqref{eq:A} and put $\mathcal{S} := \mathcal{Q}_{em} \cap \mathcal{U}$. The following assumption has an important impact on robustness properties of the superhedge. 

\begin{assumption}
In \eqref{eq:kramkov1} we assume that $\bar{C}_{m}(t) = \esssup_{\Q \in \mathcal{S}} \mathbb{E}_{\mathbb{Q}}\left[h(S_{m}^{t,x}(T))\Big|\mathcal{F}_{t}\right]$. 
\label{ass_markov_C}
\end{assumption}
\noindent

\noindent
Assumption \ref{ass_markov_C} enables us to write 
\begin{equation}
\bar{C}_{m}(t) = C_{m}^{\mathcal{S}}(t,S_{m}^{x}(t))
\label{eq:markov_eq}
\end{equation}
for some deterministic function $C_{m}^{\mathcal{S}}: [0,T] \times \mathbb{R}_{+} \to \mathbb{R}.$ Corollary \ref{cor_convex} implies that $C_{m}^{\mathcal{S}}$ is convex and has a bounded one-sided derivative in the second component. \par
 We remark that the superhedging price process is \emph{not} a special case induced by the robust pricing operator $C_{m}$ in \eqref{eq:gen_C_def} because 
$$\mathcal{M} \cap \mathcal{U} \subseteq \mathcal{Q} \cap \mathcal{U} \subseteq \mathcal{Q}_{em} \cap \mathcal{U} = \mathcal{S},$$
\noindent
and the set $\mathcal{S}$ in \eqref{eq:markov_eq} is in general larger than the other ones by not being restricted to those EMMs satisfying \emph{Condition (I)} (cf. Definition \ref{def_mart_Q}). The following theorem is a particular consequence of the assumption of the Markov property on the superhedging price. 
\begin{theorem}
Consider the process $S_{m}^{x}$ given by \eqref{eq:m_stock}, suppose Assumption \ref{ass_martingale}, Assumption \ref{ass_vol_2}, Assumption \ref{ass_markov_C} and Assumption \ref{ass_ito_C} (with $C_{m}$ replaced by $C_{m}^{\mathcal{S}}$) are satisfied. Let $\tilde{\gamma}(t,s,z) \neq 0$ for all $(t,s,z) \in [0,T] \times \mathbb{R}_{+} \times \I.$ 

\begin{enumerate}[(i)]
	\item 
	If $\gamma(t,s) > 0$ for all $(t,s) \in [0,T] \times \mathbb{R}_{+}$,
	then the superhedging price function is linear in the second variable, i.e., 
	\begin{equation}
	C_{m}^{\mathcal{S}}(t,S_{m}^{x}(t)) = a(t) + b(t) S_{m}^{x}(t),
	\label{eq:C_m_lin}
	\end{equation}
	for some deterministic functions of time $a$ and $b.$
	\item 
	Suppose $\gamma \equiv 0.$ 
	If for every $(t,s) \in [0,T] \times \mathbb{R}_{+}$ there exist $z_{1},z_{2} \in \I, z_{1} \neq z_{2}$, in the support of $\vartheta$ such that $\tilde{\gamma}(t,s,z_{1}) \neq \tilde{\gamma}(t,s,z_{2})$, then the superhedging price function is linear in the second variable taking the same form as in \eqref{eq:C_m_lin}. 
\end{enumerate}
\label{thm_superhedge}
\end{theorem}
\begin{proof}The proof is conducted in three steps: \\

\noindent
\textsc{Step 1} is to show that if $g:\mathbb{R}_{+} \to \mathbb{R}$ is convex and differentiable, $\tilde{\gamma}(t,s,z) \neq 0$ for every $(t,s,z) \in [0,T] \times \mathbb{R}_{+} \times \I,$ and  for fixed $(t,z) \in [0,T] \times \I$ it holds for every $s \in \mathbb{R}_{+}$ that
$$g(s+s\tilde{\gamma}(t,s,z)) = g(s) + s\tilde{\gamma}(t,s,z) g'(s),$$
then $g$ is linear. \\
The previous equation states that the tangent in the point $s$ is revisited by $g$ in $s+s\tilde{\gamma}(t,s,z).$ Since $g$ is convex, its graph lies above all of its tangents. Thus, we conclude that $g$ is equal to its tangent in $s$ on the interval $[s,s+s\tilde{\gamma}(t,s,z)].$ Since $s$ has been arbitrarily chosen, we see that $g$ is linear on each such interval. We are left arguing that $g$ is linear on its entire domain, but this is immediate because $\tilde{\gamma}(t,s,z) \neq 0$ for all $(t,s,z) \in [0,T] \times \mathbb{R}_{+} \times \I$ and $s + s\tilde{\gamma}(\cdot,s,\cdot) \to 0$ for $s \to 0$ and $s+s\tilde{\gamma}(\cdot,s,\cdot) \to \infty$ for $s \to \infty$ as $\tilde{\gamma}$ is bounded and strictly larger than $-1$. \\ 

\noindent

 \noindent
\textsc{Step 2} is to show that if $g:\mathbb{R}_{+} \to \mathbb{R}$ is convex and differentiable, $\tilde{\gamma}(t,s,z) \neq 0$ for all $(t,s,z) \in [0,T] \times \mathbb{R}_{+} \times \I$, and for every $z_{1},z_{2} \in \I$ and $(t,s) \in [0,T] \times \mathbb{R}_{+}$ such that $\tilde{\gamma}(t,s,z_{1}) \neq \tilde{\gamma}(t,s,z_{2})$ it holds that
\begin{equation}
\frac{g(s+s\tilde{\gamma}(t,s,z_{1}))-g(s)}{s\tilde{\gamma}(t,s,z_{1})} = \frac{g(s+s\tilde{\gamma}(t,s,z_{2}))-g(s)}{s\tilde{\gamma}(t,s,z_{2})},
\label{eq:cond_C_conv}
\end{equation}
then $g$ is linear.\\
Assume w.l.o.g. $0 < \gamma(t,s,z_{1}) < \gamma(t,s,z_{2}).$ An easy calculation allows us to deduce from \eqref{eq:cond_C_conv} that 
\begin{equation}
\frac{g(s+s\tilde{\gamma}(t,s,z_{1}))-g(s)}{s\tilde{\gamma}(t,s,z_{1})} = \frac{g(s+s\tilde{\gamma}(t,s,z_{2}))-g(s+s\tilde{\gamma}(t,s,z_{1}))}{s(\tilde{\gamma}(t,s,z_{2})-\tilde{\gamma}(t,s,z_{1}))}.
\label{eq:cond_C_conv_2}
\end{equation}
Observe that
\begin{align*}
\frac{g(s+s\tilde{\gamma}(t,s,z_{1}))-g(s)}{s\tilde{\gamma}(t,s,z_{1})} &= \frac{\int_{s}^{s+s\tilde{\gamma}(t,s,z_{1})}g'(u)\ du}{s\tilde{\gamma}(t,s,z_{1})} \\
& \leq g'(s+s\tilde{\gamma}(t,s,z_{1})) \\
& \leq \frac{\int_{s+s\tilde{\gamma}(t,s,z_{1})}^{s+s\tilde{\gamma}(t,s,z_{2})}g'(u)\ du}{s(\tilde{\gamma}(t,s,z_{2})-\tilde{\gamma}(t,s,z_{1}))} = \frac{g(s+s\tilde{\gamma}(t,s,z_{2}))-g(s+s\tilde{\gamma}(t,s,z_{1}))}{s(\tilde{\gamma}(t,s,z_{2})-\tilde{\gamma}(t,s,z_{1}))},
\end{align*}
where the two inequalities follow because $g$ is convex. As \eqref{eq:cond_C_conv_2} states that all the terms in the previous chain of (in-)equalities are equal, we conclude that
\begin{align*}
g'(u) &= g'(s+s\tilde{\gamma}(t,s,z_{1})) \ \text{for all}\ u \in \ [s,s+s\tilde{\gamma}(t,s,z_{1})], \\
g'(u) &= g'(s+s\tilde{\gamma}(t,s,z_{1})) \ \text{for all}\ u \in \ [s+s\tilde{\gamma}(t,s,z_{1}),s+s\tilde{\gamma}(t,s,z_{2})],
\end{align*}
i.e., $g$ is linear on the interval (the upper and the lower bound of the interval might switch depending on the sign of $\tilde{\gamma}$). Referring to \textsc{Step 1}, one can show by analogous arguments that $g$ is linear on its entire domain.   \\

\noindent
\textsc{Step 3} is to establish the claim of the theorem. An application of \emph{It\^{o}'s lemma} to $C_{m}^{\mathcal{S}}$ yields

\begin{align}
&C_{m}^{\mathcal{S}}(t,S_{m}^{x}(t))\notag \\
\begin{split}
&= C_{m}^{\mathcal{S}}(0,x) + \int_{0}^{t} \dot{C}_{m}^{\mathcal{S}}(u,S_{m}^{x}(u))\ du + \int_{0}^{t} C_{m}^{\mathcal{S}\prime}(u,S_{m}^{x}(u\m))\ dS_{m}^{x}(u)\\
&\ \ + \frac{1}{2} \int_{0}^{t} C_{m}^{\mathcal{S}\prime \prime}(u,S_{m}^{x}(u))S_{m}^{x}(u)^{2} \gamma(u,S_{m}^{x}(u))^{2}\ du \\
&\ \ + \int_{0}^{t}\int_{\I} \Big(C_{m}^{\mathcal{S}}(u,S_{m}^{x}(u\m) + \Delta S_{m}^{x}(u,z)) - C_{m}^{\mathcal{S}}(u,S_{m}^{x}(u\m))\\
&\ \ \phantom{+ \int_{0}^{t}\int_{\I} \Big(} - \Delta S_{m}^{x}(u,z) C_{m}^{\mathcal{S}\prime}(u,S_{m}^{x}(u\m))  \Big)\ \tilde{J}(du,dz)\\
&\ \  + \int_{0}^{t} \int_{\I}\Big(C_{m}^{\mathcal{S}}(u,S_{m}^{x}(u\m) + \Delta S_{m}^{x}(u,z))\\
&\ \ \phantom{+ \int_{0}^{t} \int_{\I}\Big(} - C_{m}^{\mathcal{S}}(u,S_{m}^{x}(u\m)) - \Delta S_{m}^{x}(u,z) C_{m}^{\mathcal{S}\prime}(u,S_{m}^{x}(u\m))  \Big)\ \vartheta(dz)du.
\end{split}
\label{eq:C}
\end{align}

\noindent
The \emph{optional decomposition theorem} (cf. Theorem 1 in \cite{follmer1997optional}) implies the existence of some predictable $S_{m}^{x}$-integrable process $\pi = (\pi(t))_{t \in [0,T]}$ and an increasing adapted process $A=(A(t))_{t \in [0,T]}$ with $A(0) = 0$ such that for every $t \in [0,T]$ we have

\begin{equation}
C_{m}^{\mathcal{S}}(t,S_{m}^{x}(t)) = C_{m}^{\mathcal{S}}(0,x) + \int_{0}^{t} \pi(u\m)\ dS_{m}^{x}(u) - A(t).  
\label{eq:aux_doob}
\end{equation}
Equating \eqref{eq:C} and \eqref{eq:aux_doob} and using that the process $A$ is of finite variation because it is increasing, the calculation of the quadratic variation yields the following implication:

%
\begingroup
\allowdisplaybreaks
\begin{align}
0 &= \int_{0}^{t} \left(C_{m}^{\mathcal{S}\prime}(u,S_{m}^{x}(u\m)) - \pi(u\m)\right)^{2}\ S_{m}^{x}(u)^{2} \gamma(u,S_{m}^{x}(u))^{2} \ du \notag \\
&\ \ + \int_{0}^{t} \int_{\I} \Bigg(\left(C_{m}^{\mathcal{S}\prime}(u,S_{m}^{x}(u\m)) - \pi(u\m)\right) S_{m}^{x}(u) \tilde{\gamma}(u,S_{m}^{x}(u),z) \notag  + C_{m}^{\mathcal{S}}(u,S_{m}^{x}(u\m) + \Delta S_{m}^{x}(u,z))\\
&\ \ \phantom{+ \int_{0}^{t} \int_{\I} \Bigg( (} - C_{m}^{\mathcal{S}}(u,S_{m}^{x}(u\m)) - \Delta S_{m}^{x}(u,z) C_{m}^{\mathcal{S}\prime}(u,S_{m}^{x}(u\m))\Bigg)^{2} \ \vartheta(dz)du \notag \\
\begin{split}
&=  \int_{0}^{t} \left(C_{m}^{\mathcal{S}\prime}(u,S_{m}^{x}(u\m)) - \pi(u\m)\right)^{2}\ S_{m}^{x}(u)^{2} \gamma(u,S_{m}^{x}(u))^{2} \ du\\
&\ \ + \int_{0}^{t} \int_{\I} \Big(C_{m}^{\mathcal{S}}(u,S_{m}^{x}(u\m) + \Delta S_{m}^{x}(u,z)) - C_{m}^{\mathcal{S}}(u,S_{m}^{x}(u\m))\\
&\ \ \phantom{+ \int_{0}^{t} \int_{\I} \Big(} -\pi(u\m) S_{m}^{x}(u) \tilde{\gamma}(u,S_{m}^{x}(u),z) \Big)^{2} \ \vartheta(dz)du. 
\end{split}
\label{eq:long_eq}
\end{align}
\endgroup
\noindent
Recall from \textsc{Step 4} of the proof of Theorem \ref{thm_convex} that the function $x \mapsto S_{m}^{x}(t)$ is differentiable in $x$ up to a set of measure zero $\mathcal{N}_{1}.$ The term in brackets in the last two lines of \eqref{eq:long_eq} is equal to zero for any fixed $\omega \in \Omega$ up to some set of measure zero $\mathcal{N}_{2}.$ Define $\mathcal{N} := \mathcal{N}_{1} \cup \mathcal{N}_{2}.$ In the sequel, we restrict to $\omega \in \Omega \setminus \mathcal{N}.$ Define $\mathbb{T}(\omega) := \{t \in [0,T]| \Delta S_{m}^{x}(t,z,\omega) \neq 0\}.$ \textsc{Step 1} of the proof of Theorem \ref{thm_main} states that the image of the mapping $x \mapsto S_{m}^{x}(t)$ is $\mathbb{R}_{+}.$ Thus, fixing $(t,s,z) \in [0,T] \times \mathbb{R}_{+} \times \I$, it must hold that

\begin{equation}
C_{m}^{\mathcal{S}}(t,s+s\tilde{\gamma}(t,s,z)) - C_{m}^{\mathcal{S}}(t,s) - s\tilde{\gamma}(t,s,z) \pi(t,s) = 0.
\label{eq:key2}
\end{equation}
\noindent
We turn to the discussion of the cases (i) and (ii): 

\begin{enumerate}[(i)]
	\item If $\gamma(t,s) \neq 0$ for all $(t,s) \in [0,T] \times \mathbb{R}_{+}$, we learn from the first line of \eqref{eq:long_eq} that \eqref{eq:key2} holds with $\pi(t,s) = C_{m}^{\mathcal{S}\prime}(t,s).$ Consequently, \textsc{Step 1} implies that $C_{m}^{\mathcal{S}}(t,\cdot)$ is linear in the second variable for every $t \in \mathbb{T}.$ Since the set of jump times  $\mathbb{T}(\omega)$ is dense in $[0,T]$ for every fixed $\omega \in \Omega \setminus \mathcal{N}$ (see \cite{ct}, p.84 for a proof)  and as $C_{m}^{\mathcal{S}}$ is continuous by Assumption \ref{ass_ito_C}, we conclude that $C_{m}^{\mathcal{S}}(t,\cdot)$ is linear in the second argument for every $t \in [0,T].$
	\\
	\item Suppose $\gamma \equiv 0.$ Let $z_{1},z_{2} \in \I,$  such that $\tilde{\gamma}(t,s,z_{1}) \neq \tilde{\gamma}(t,s,z_{2}).$ By continuity of $\tilde{\gamma}$ and \eqref{eq:key2} it must hold that 
	
	$$\frac{C_{m}^{\mathcal{S}}(t,s+ s\tilde{\gamma}(t,s,z_{1}))-C_{m}^{\mathcal{S}}(t,s)}{s\tilde{\gamma}(t,s,z_{1})} = \frac{C_{m}^{\mathcal{S}}(t,s+ s\tilde{\gamma}(t,s,z_{2}))-C_{m}^{\mathcal{S}}(t,s)}{s\tilde{\gamma}(t,s,z_{2})}.$$
	\noindent
	We immediately see that \textsc{Step 2} implies that $C_{m}^{\mathcal{S}}(t,\cdot)$ must be linear in the second component for every $t \in \mathbb{T}.$ By the same argument as in (i) we conclude linearity for every $t \in [0,T].$ 

\end{enumerate}
\end{proof}

\begin{remark}
We remark that Theorem \ref{thm_superhedge} is a robustness result in incomplete markets upon assuming that the superhedging price process is Markov. A perfect hedge is possible due to the induced linearity of the payoff function. Thus, robustness is trivial.  
\end{remark}

Finally, we broach the issue of robustness in a complete market when the stock price exhibits jumps. The following corollary treats the case $\gamma \equiv 0$ and $\vartheta(dz) = \lambda\ \mathbbm{1}_{\{\alpha\}}(dz)$ for some $\lambda > 0$ and $\alpha \in \I.$ Hence, the jump component corresponds to a homogeneous Poisson process. In this case the risk-neutral measure is unique and consequently the set $\mathcal{Q}_{em}$ is a singleton containing $\mathbb{P}$ solely. To make this explicit, we denote the contingent claim price at time $t$ by $C_{m}^{\mathbb{P}}(t,S_{m}^{x}(t)).$  
\begin{corollary}
Suppose all assumptions of Theorem \ref{thm_superhedge} are satisfied, $\gamma \equiv 0$ and $\vartheta(dz) = \lambda\ \mathbbm{1}_{\{\alpha\}}(dz)$ for some $\lambda > 0$ and $\alpha \in \I.$ The replicating Delta-strategy in the misspecified model is given by $\pi = (\pi(t,S_{m}^{x}(t)))_{t \in [0,T]}$ with  
$\pi(t,S_{m}^{x}(t)) = \frac{C_{m}^{\mathbb{P}}(t,S_{m}^{x}(t\m) + \Delta S_{m}^{x}(t,\alpha))-C_{m}^{\mathbb{P}}(t,S_{m}^{x}(t\m))}{\Delta S_{m}^{x}(t,\alpha)}.$
 If in addition
$\emph{sgn}(\tilde{\gamma}(t,\tilde{S}_{m}^{x}(t),\alpha)-\eta(t,\alpha)) = \emph{sgn}(\eta(t,\alpha)),$
$d\mathbb{P} \times dt$-a.s., then following this strategy and trading in the real stock $S$ yields an a.s. super-replication of $h(S(T))$, i.e., 
$$C_{m}^{\mathbb{P}}(0,x) + \int_{0}^{T} \pi(u,S(u\m))\ dS(u) \geq h(S(T)).$$
\label{cor_poisson}
\end{corollary}
\begin{proof}
We start with \eqref{eq:C} and use \eqref{eq:aux_doob} to prove the corollary. Since the market is complete, the claim $h(S_{m}^{x}(T))$ is perfectly replicable and the process $A$ in \eqref{eq:aux_doob} (allowing for the interpretation of the cash amount that can be withdrawn at time $t$) is then (cf. \eqref{eq:price1}) equal to zero, i.e., 
\begin{align*}
A(t) &= \int_{0}^{t} \dot{C}_{m}^{\mathbb{P}}(u,S_{m}^{x}(u))\ du +\lambda \int_{0}^{t} \Big(C_{m}^{\mathbb{P}}(u,S_{m}^{x}(u\m) + \Delta S_{m}^{x}(u,\alpha)) - C_{m}^{\mathbb{P}}(u,S_{m}^{x}(u\m))\\
&\ \ \phantom{\int_{0}^{t} \dot{C}_{m}^{\mathbb{P}}(u,S_{m}^{x}(u))\ du +\lambda \int_{0}^{t} \Big(}\ - \Delta S_{m}^{x}(u,\alpha) C_{m}^{\mathbb{P}'}(u,S_{m}^{x}(u\m))  \Big)\  du  =0,
\end{align*}
\noindent
for Lebesgue-almost all $t \in [0,T].$ Then the \emph{Lebesgue differentiation theorem} implies that \eqref{eq:A_gen} holds with equality:
\begin{align}
\begin{split}
0 &= \dot{C}_{m}^{\mathbb{P}}(t,S_{m}^{x}(t)) +\lambda \Big(C_{m}^{\mathbb{P}}(t,S_{m}^{x}(t) + S_{m}^{x}(t)\tilde{\gamma}(t,S_{m}^{x}(t),\alpha)) - C_{m}^{\mathbb{P}}(t,S_{m}^{x}(t))\\
&\ \ \phantom{\dot{C}_{m}^{\mathbb{P}}(t,S_{m}^{x}(t)) +\lambda \Big(}\ - S_{m}^{x}(t)\tilde{\gamma}(t,S_{m}^{x}(t),\alpha) C_{m}^{\mathbb{P}'}(t,S_{m}^{x}(t))  \Big).  
\end{split}
\label{eq:LD}
\end{align}
We can replace $S_{m}^{x}(t)$ in \eqref{eq:LD} by the the true stock price $S(t)$ because \eqref{eq:LD} holds for all $s \in \mathbb{R}_{+}$ (see \textsc{Step 2} of the proof of Theorem \ref{thm_main} for detailed arguments).  Since $C_{m}^{\mathbb{P}}(T,S(T)) = h(S(T)),$ trading in the physical stock $S$ yields
\begingroup
\allowdisplaybreaks
\begin{align*}
\begin{split}
&C_{m}^{\mathbb{P}}(0,x) + \int_{0}^{T}\left( \frac{C_{m}^{\mathbb{P}}(u,S(u\m) + \Delta S(u,\alpha))-C_{m}^{\mathbb{P}}(u,S(u\m))}{\Delta S(u,\alpha)}\right)\ dS(u)\notag \\
&=C_{m}^{\mathbb{P}}(0,x) + \int_{0}^{T} C_{m}^{\mathbb{P}'}(u,S(u\m))\ dS(u)\\
&\ \  + \int_{0}^{T} \Big(C_{m}^{\mathbb{P}}(u,S(u\m) + \Delta S(u,\alpha)) - C_{m}^{\mathbb{P}}(u,S(u\m)) - \Delta S(u,\alpha) C_{m}^{\mathbb{P}'}(u,S(u\m))  \Big)\ \tilde{J}(du,dz)\end{split}  \\
& = h(S(T)) - \Bigg( \int_{0}^{T} \dot{C}_{m}^{\mathbb{P}}(u,S(u))\ du \\
&\ \ +\lambda \int_{0}^{T} \Big(C_{m}^{\mathbb{P}}(u,S(u\m) + \Delta S(u,\alpha)) - C_{m}^{\mathbb{P}}(u,S(u\m)) - \Delta S(u,\alpha) C_{m}^{\mathbb{P}'}(u,S(u\m)) \Big)\  du \Bigg)\notag  \\
&\geq h(S(T)) - \Bigg( \int_{0}^{T} \dot{C}_{m}^{\mathbb{P}}(u,S(u)) du +\lambda\int_{0}^{T} \Big(C_{m}^{\mathbb{P}}(u,S(u\m) + S(u\m) \tilde{\gamma}(u,S(u\m),\alpha))\notag  - C_{m}^{\mathbb{P}}(u,S(u\m))\\
&\ \ \phantom{h(S(T)) - \Bigg( \int_{0}^{T} \dot{C}_{m}^{\mathbb{P}}(u,S(u)) du +\lambda\int_{0}^{T} \Big(} - S(u\m) \tilde{\gamma}(u,S(u\m),\alpha) C_{m}^{\mathbb{P}'}(u,S(u\m)) \Big)\ du \Bigg)= h(S(T)), \notag
\end{align*}
\endgroup
where the second equality is justified by \emph{It\^{o}'s formula} and the last inequality follows from the convexity of $C_{m}^{\mathbb{P}}$ in the second component (cf. Corollary \ref{cor_convex}) combined with Lemma \ref{lem_mon}. The final equality is immediate from \eqref{eq:LD}. 
\end{proof}

\setcounter{section}{0}
\setcounter{theorem}{0}
\setcounter{equation}{0}
\renewcommand{\theequation}{\thesection.\arabic{equation}}
\section*{Appendix}
\renewcommand{\thesection}{A}

\begin{lemma}
	\label{lemmagirsanov}
Consider the process $S_{m}^{x}$ given by \eqref{eq:m_stock} w.r.t. the reference measure $\mathbb{P}.$ Then it holds that $\mathbb{Q}<<\p$ (i.e., $\Q$ is absolutely continuous with respect to $\p$) if and only if there exist predictable processes $\psi: [0,T] \rightarrow \mathbb{R}$ and $\theta: [0,T] \times  \I \rightarrow \mathbb{R}$ with $\theta(t,z) < 1,$  
such that the process $\xi = (\xi(t))_{t \in [0,T]}$ with
\begin{align*}
\xi(t) &:= \Big(-\int_{0}^{t} \psi(u) \ dW(u) - \frac{1}{2} \int_{0}^{t} \psi(u)^{2} \ du \\
&\ \ + \int_{0}^{t} \int_{\I} \ln(1-\theta(u,z)) \ \tilde{J}(du,dz)  + \int_{0}^{t} \int_{\I} \left(\ln(1-\theta(u,z)) + \theta(u,z)\right)\ \vartheta(dz)du \Big)
\end{align*}
is a well-defined martingale with $\xi_t=\mathbb{E}_{\mathbb{P}}[\frac{d\Q}{d\p}|\mathcal{F}_t] .$ Define
\begin{align}
\begin{split}
dW_{\mathbb{Q}}(t) &= \psi(t) \ dt + dW(t), \\
\tilde{J}_{\mathbb{Q}}(dt,dz) &= \theta(t,z) \ \vartheta(dz)dt + \tilde{J}(dt,dz).
\end{split}
\label{eq:girsanov}
\end{align} 
Then $W_{\mathbb{Q}}$ is a standard Brownian motion w.r.t $\mathbb{Q}$ and $\tilde{J}_{\mathbb{Q}}(dt,dz)$ is the $\mathbb{Q}$-compensated version of the Poisson random measure $J(dt,dz).$ Finally, $\Q\sim \p$ is in $ \mathcal{Q}_{em}$ if and only if
\begin{equation}
\gamma(t,\tilde{S}_{m}^{x}(t)) \psi(t) + \int_{\I} \tilde{\gamma}(t,\tilde{S}_{m}^{x}(t\m),z) \theta(t,z) \ \vartheta(dz) = 0, \quad d\mathbb{P} \times dt\,\,a.s.
\label{eq:cond_martingale}
\end{equation} 
\label{lemma_girsanov}
\end{lemma}
\begin{proof}
See for instance \cite{oksjump}, Chapter 1.4.
\end{proof}
Hence, every probability measure $\Q$ which is absolutely continuous with respect to $\p$ induces a pair $(\phi,\theta)$. On the other, it follows for instance by \cite{kazamaki} that a pair $(\psi,\theta)$ induces a probablity measure $\Q^{(\psi,\theta)}$ with Radon-Nikodym derivative $\xi_T$ defined above if there exists $\bar{\delta}>0$ such that $\theta<1-\bar{\delta}$, and  $\psi$ and $\theta$ are BMO processes, meaning that there exists $C>0$ such that for all $t$
$$ \mathbb{E}\bigg[\int_t^T\Big[ |\psi(s)|^2 +\int_{\I} |\theta(s,z)|^2\nu (dz) \Big]ds\bigg|\mathcal{F}_t\bigg]\leq C.
$$

\footnotesize
\bibliographystyle{apa}
\bibliography{bib}
\footnotesize

\end{document}